\documentclass[onecolumn, 11pt]{IEEEtran}

\usepackage{geometry}
\usepackage[shortlabels]{enumitem}
\usepackage{setspace, dsfont}
\usepackage{amssymb, amsmath, mathrsfs, amsthm, stmaryrd, upgreek, mathtools,yfonts,bbm,bm}
\usepackage{pgfplots, caption, tikz, subfigure, graphicx}
\usepackage{url, hyperref, cite}
\usepackage{MyMnSymbol}
\usepackage{xcolor}
\usepackage{calc}

\input{def.tex}

\begin{document}

\title{Beurling-Type Density Criteria for \\[-0.3cm] System Identification\\[0.5cm]}

\author{\IEEEauthorblockN{Verner Vla\v{c}i\'c\IEEEauthorrefmark{1}, C\'eline Aubel\IEEEauthorrefmark{2}, and Helmut B\"olcskei\IEEEauthorrefmark{1} \\[0.2cm]}
\IEEEauthorblockA{\IEEEauthorrefmark{1}ETH Zurich, Switzerland\\
\IEEEauthorrefmark{2}Swiss National Bank, Zurich, Switzerland\\
Email: \IEEEauthorrefmark{1}vlacicv@mins.ee.ethz.ch, \IEEEauthorrefmark{2}celine.aubel@gmail.com, \IEEEauthorrefmark{1}hboelcskei@ethz.ch}
}

\maketitle

\newcommand{\norm}[1]{\left\|#1\right\|}
\newcommand{\bargmann}{\textfrak{B}\hspace{1pt}}
\newcommand{\BFSp}[2]{\m{F}^{#1}(#2)}
\newcommand{\linOp}{\m{H}}

\newcommand{\sep}{\mrm{sep}}
\DeclarePairedDelimiter{\ceil}{\lceil}{\rceil}

\newcommand{\MSp}[2]{M^{#1}(#2)}
\newcommand{\moduPair}[2]{M^{#1}(\R)\times M^{#2}(\R)}
\newcommand{\modu}{\m{M}}
\newcommand{\tra}{\m{T}}
\newcommand{\wtd}{\widetilde}
\newcommand{\coleqq}{\vcentcolon=}
\newcommand{\m}{\mathcal}
\newcommand{\mrm}{\mathrm}
\newcommand{\msc}{\mathscr}
\newcommand{\bve}{\bm{\varepsilon}}
\renewcommand{\Re}{\mrm{Re}}
\renewcommand{\Im}{\mrm{Im}}

\theoremstyle{definition}
\newtheorem{definition}{Definition}

\begin{abstract}
	This paper addresses the problem of identifying a linear time-varying (LTV) system characterized by a (possibly infinite) discrete set of delay-Doppler shifts without a lattice (or other ``geometry-discretizing'') constraint on the support set. Concretely, we show that a class of such LTV systems is identifiable whenever the upper uniform Beurling density of the delay-Doppler support sets, measured ``uniformly over the class'', is strictly less than $1/2$. The proof of this result reveals an interesting relation between LTV system identification and interpolation in the Bargmann-Fock space. Moreover, we show that this density condition is also necessary for classes of systems invariant under time-frequency shifts and closed under a natural topology on the support sets. 
We furthermore show that identifiability guarantees robust recovery of the delay-Doppler support set, as well as the weights of the individual delay-Doppler shifts, both in the sense of asymptotically vanishing reconstruction error for vanishing measurement error.
\end{abstract}


\section{Introduction}

Identification of deterministic linear time-varying (LTV) systems has been a topic of long-standing interest, dating back to the seminal work by Kailath~\cite{Kailath1962} and Bello~\cite{Bello1969}, and has seen significant renewed interest during the past decade~\cite{Kozek2005,Heckel2013,Bajwa2011,Heckel2014}.
This general problem occurs in many fields of engineering and science. Concrete examples include system identification in control theory and practice, the measurement of dispersive communication channels, and radar imaging. The formal problem statement is as follows. We wish to identify the LTV system $\linOp$ from its response
\begin{equation}
	(\linOp x )(t) \coleqq \int_{\R^2} S_{\linOp}(\tau,\nu) \, x (t - \tau) \, e^{2\pi i\nu t} \,\mrm{d}\tau\mrm{d}\nu,\quad \forall t \in \R,
	\label{eq:ltvsys}
\end{equation}
to a probing signal $x(t)$, with $S_{\linOp}(\tau,\nu)$ denoting the spreading function associated with $\linOp$. Specifically, we consider $\linOp$ to be identifiable if there exists an $x$ such that knowledge of $\linOp x$ allows us to determine $S_{\linOp}$. The representation theorem \cite[Thm. 14.3.5]{Groechenig2000} states that a large class of continuous linear operators can be represented as in \eqref{eq:ltvsys}. 

Kailath \cite{Kailath1962} showed that an LTV system with spreading function supported on a rectangle centered at  the origin of the $(\tau,\nu)$-plane is identifiable if the area of the rectangle is at most $1$. This result was later extended by Bello to arbitrarily fragmented spreading function support regions with the support area measured collectively over all supporting pieces \cite{Bello1969}. Necessity of the Kailath-Bello condition was established in \cite{Kozek2005,Pfander2006} through elegant functional-analytic arguments. 
However, all these results require the support region of $S_{\linOp}(\tau,\nu)$ to be known prior to identification, a condition that is very restrictive and often impossible to realize in practice. More recently, it was demonstrated in \cite{Heckel2013} that identifiability in the more general case considered by Bello \cite{Bello1969} is possible without prior knowledge  of the spreading function support region, again as long as its area (measured collectively over all supporting pieces) is upper-bounded by $1$. This is surprising as it says that there is no price to be paid for not knowing the spreading function's support region in advance. The underlying insight has 
strong conceptual ties to the theory of spectrum-blind sampling of sparse multi-band signals 
\cite{Feng1996,Feng1997,Lu2008,Mishali2009}.

The situation is fundamentally different when the spreading function is discrete according to
\begin{equation}
	(\linOp x )(t) \coleqq \sum_{m\in\N} \alpha_{m} \, x (t - \tau_m) \, e^{2\pi i\nu_m t},\quad \forall t \in \R,
	\label{eq: radio communication channel}
\end{equation}
where $(\tau_m, \nu_m)\in\R^2$ are delay-Doppler shift parameters and $\alpha_m$ are the corresponding complex weights, for $m\in\N$. Here, the (discrete) spreading function can be supported on unbounded subsets of the $(\tau,\nu)$-plane with the identifiability condition on the support area of the spreading function replaced by a density condition on the support set $\supp(\linOp)\coleqq \{(\tau_m, \nu_m):m\in\N\}$. Specifically, for $\linOp$ supported on rectangular lattices according to $\supp(\linOp)=a^{-1}\Z \times b^{-1}\Z$, Kozek and Pfander established that $\linOp$ is identifiable if and only if $ab \leq 1$ \cite{Kozek2005}.
In \cite{Grip2013} a necessary condition for identifiability of a set of Hilbert-Schmidt operators defined analogously to \eqref{eq: radio communication channel} is given; this condition is expressed in terms of the Beurling density of the support set, but the time-frequency pairs $(\tau_m,\nu_m)$ are assumed to be confined to a lattice. 
Now, in practice the discrete spreading function will not be supported on a lattice as the parameters $\tau_m,\nu_m$ correspond to time delays and frequency shifts induced, e.g. in wireless communication, by the propagation environment. It is hence of interest to understand the limits on identifiability in the absence of ``geometry-discretizing" assumptions---such as a lattice constraint---on $\supp(\linOp)$. Resolving this problem is the aim of the present paper.

\subsection{Fundamental limits on identifiability}

The purpose of this paper is twofold. First, we establish fundamental limits on the stable identifiability of $\linOp$ in \eqref{eq: radio communication channel} in terms of $\supp(\linOp)$ and $\{\alpha_{m}\}_{m\in\N}$. Our approach is based on the following insight. Defining the discrete complex measure $\mu \coleqq \sum_{m\in\N} \alpha_{m}\, \delta_{\tau_m, \nu_m}$ on $\R^2$, where $\delta_{\tau_m,\nu_m}$ denotes the Dirac point measure with mass at $(\tau_m,\nu_m)$, the input-output relation \eqref{eq: radio communication channel} can be formally rewritten as
\begin{equation}\label{eq:measure-int-proto}
	\quad (\linOp_\mu x )(t) = \int_{\R^2} x (t - \tau)e^{2\pi i\nu t}\, \mrm{d}{\mu}(\tau, \nu),\quad t\in\R,
\end{equation}
where we use throughout $\linOp_\mu$ instead of $\linOp$ for concreteness.
Identifying the system $\linOp_\mu$ thus amounts to reconstructing the discrete measure $\mu$ from $\linOp_\mu x $. More specifically, we wish to find necessary and sufficient conditions on classes $\msc{H}$ of measures guaranteeing stable identifiability bounds of the form
\begin{equation}\label{eq:ident-proto}
d_{\text{r}}(\mu,\mu')\leq d_{\text{m}}(\linOp_{\mu} x , \linOp_{\mu'} x )
, \quad \text{for all }\mu,\mu'\in\msc{H},
\end{equation}
for appropriate reconstruction and measurement metrics $d_{\text{r}}$ and $d_{\text{m}}$, where $\mu$ is the ground truth measure to be recovered and $\mu'$ is the estimated measure. The class $\msc{H}$ can be thought of as modelling the prior information available about the measure $\mu$ facilitating its identification by restricting the set of potential estimated measures $\mu'$. In particular, the smaller the class $\msc{H}$, the ``easier'' it should be to satisfy \eqref{eq:ident-proto}. In addition to the class $\msc{H}$ of measures itself, the existence of a bound of the form \eqref{eq:ident-proto} depends on the choice of the probing signal $x$, so we will later speak of \emph{identifiability by $x$}.

This formulation reveals an interesting connection to the super-resolution problem as studied by Donoho~\cite{Donoho1991}, where the goal is to recover a discrete complex measure on $\R$, i.e., a weighted Dirac train, from low-pass measurements. The problem at hand, albeit formally similar, differs in several important aspects. First, we want to identify a measure $\mu$ on $\R^2$, i.e., a measure on a \emph{two-dimensional} set, from observations in \emph{one} parameter, namely $(\linOp_\mu x )(t)$, $t\in\R$. Next, the low-pass observations in~\cite{Donoho1991} are replaced by short-time Fourier transform-type observations, where the probing signal $x $ appears as the window function. While super-resolution from STFT-measurements was considered in \cite{Aubel2017}, the underlying measure to be identified in \cite{Aubel2017} is, as in \cite{Donoho1991}, on $\R$. 
Finally,  \cite{Donoho1991}  assumes that the support set of the measure under consideration is confined to an a priori fixed lattice. While such a strong structural assumption allows for the reconstruction metric $d_{\text{r}}$ to take a simple and intuitive form, it unfortunately bars taking into account the geometric properties of the support sets considered. By contrast, the general definition of stable identifiability (see Definition \ref{defn: identifiability}) analogous to \cite{Donoho1991} will pave the way for a theory of support recovery without a lattice assumption, as discussed in the next subsection.

These differences make for very different technical challenges. Nevertheless, we can follow the spirit of Donoho's work~\cite{Donoho1991}, who established necessary and sufficient conditions for stable identifiability in the classical super-resolution problem. Donoho's conditions are expressed in terms of the uniform Beurling density of the measure's (one-dimensional) support set and are derived using density theorems for interpolation in the Bernstein and Paley-Wiener spaces~\cite{Beurling1989-2} and for the balayage of Fourier-Stieltjes transforms~\cite{Beurling1989-1}.
We will, likewise, establish a sufficient condition guaranteeing stable identifiability for classes of measures whose supports have density less than 1/2 ``uniformly over the class $\msc{H}$'' (formally introduced in Definition \ref{def:UBcd}). In addition, we show that this is also a necessary condition for classes of measures invariant under time-frequency shifts and closed under a natural topology on the support sets. We will see below that these requirements are not very restrictive as we present several examples of identifiable and non-identifiable classes of measures.
The proofs of these results are based on the density theorem for interpolation in the Bargmann-Fock space~\cite{Seip1992, Seip1992-1, Seip1992-2, Seip1992-3}, as well as several results about Riesz sequences from \cite{Groechenig2015}.

\subsection{Robust recovery of the delay-Doppler support set}

The second goal of the paper is to address the implications of the identifiability condition on the recovery of the discrete measure $\mu$. Concretely, suppose that we want to recover a fixed measure  $\mu \coleqq \sum_{m\in\N} \alpha_{m}\, \delta_{\tau_m, \, \nu_m}$ from a known class of measures $\msc{H}$ assumed to be stably identifiable (in the sense of \eqref{eq:ident-proto}) with respect to a probing signal $x$, and let $\{\mu_n\}_{n\in\N}\subset \msc{H}$ be a sequence of ``estimated candidate measures'' $\mu_n \coleqq \sum_{m\in\N} \alpha_{m}^{(n)}\, \delta_{\tau_m^{(n)},\, \nu_m^{(n)}}$ for the recovery of $\mu$. We will show that, under a mild regularity condition on $x$, the stable identifiability condition \eqref{eq:ident-proto} on $\msc{H}$ guarantees that
\begin{equation}\label{eq:proto-robust-intro}
\linOp_{\mu_n} x \to \linOp_{\mu} x \quad \implies \quad \supp(\linOp_{\mu_n})\to  \supp(\linOp_\mu) \quad \text{and}\quad \{\alpha_{m}^{(n)}\}_{m\in\N}\to  \{\alpha_{m}\}_{m\in\N},
\end{equation}
as $n\to\infty$, where the topologies in which these limits take place will be specified in due course. In words, this result says that the better the measurements $\linOp_{\mu_n}$ match the true measurement $\linOp_{\mu}$, the closer the estimated measures $\mu_n$ are to the ground truth $\mu$. This, in particular, shows that ``measurement matching'' is sufficient for recovery within stably identifiable classes $\msc{H}$, i.e., any algorithm that generates a sequence of measures $\{\mu_n\}_{n\in\N}\subset \msc{H}$ satisfying $\linOp_{\mu_n} x \to \linOp_{\mu} x$ will succeed in recovering $\mu\in \msc{H}$. Crucially, we do not assume that the support sets $ \supp(\linOp_{\mu})$ and $ \supp(\linOp_{\mu_n})$, for $n\in\N$, are confined to a lattice (or any other a priori fixed discrete set). To the best of our knowledge, this is the first known LTV system identification result on the robust recovery of the discrete support set of the measure, instead of its weights only.

\textit{Notation.}
We write $B_R(a)$ for the closed ball in $\C$ of radius $R$ centered at $a$, and denote its boundary by $\partial B_R(a)$.  For a set $S\subset \C$, we let $\mathds{1}_S:\C\to \R$ be the indicator function of $S$, taking on the value $1$ on $S$ and $0$ elsewhere. 
We will identify $\C$ with $\R^2$ whenever appropriate and convenient.

We say that a set $\Lambda\subset \C$ is discrete if, for all $\lambda \in \Lambda$, one can find a $\delta > 0$ such that $|\lambda - \lambda'| > \delta$, for all $\lambda' \in \Lambda \!\setminus \{\lambda\}$.
Following the terminology employed in \cite[\S 2.2]{Groechenig2015}, we say that a set $\Lambda\subset\C$ is \textit{relatively separated} if
\begin{equation*}
\mrm{rel}(\Lambda)\coleqq\sup{\{\#(\Lambda\cap B_1(x)):x\in\C\}}<\infty.
\end{equation*}
Further, we say that $\Lambda$ is \textit{separated} (usually referred to as \textit{uniformly discrete} in the literature), if
\begin{equation*}
\sep(\Lambda)\coleqq\inf\{|\lambda - \lambda'| \colon \lambda, \lambda' \in \Lambda, \lambda \neq \lambda'\} > 0.
\end{equation*}
Finally, for two separated sets $\Lambda_1,\Lambda_2\subset \R^2$, we define their mutual separation according to
\begin{equation}\label{2setsMinSep}
\mrm{ms}(\Lambda_1,\Lambda_2)\coleqq\inf_{\substack{\lambda_1\in\Lambda_1,\lambda_2\in\Lambda_2\\\lambda_1\neq \lambda_2}}{| \lambda_1-\lambda_2 |}.
\end{equation}
Note that points that are elements of both $\Lambda_1$ and $\Lambda_2$ are excluded from consideration in the expression for mutual separation.

For a Banach space $\m{B}$, we write  $\norm{\cdot}_{\m{B}}$, $\m{B}^*$, and  $\innerProd{\cdot}{\cdot}_{\m{B}\times\m{B}^*}$ to denote the norm, the topological dual of $\m{B}$, and the dual pairing on $\m{B}$, respectively. 
Throughout the paper we use $p$ and $q$ to denote conjugate indices in $[1,\infty]$ such that $1/p+1/q=1$. We write $\msc{M}^p$ for the vector space of all complex Radon measures on $\C$ of the form $\mu=\sum_{\lambda\in\Lambda}\alpha_\lambda \delta_\lambda$, where $\Lambda$ is a relatively separated discrete subset of $\R^2$, $\{\alpha_\lambda\}_{\lambda\in\Lambda}$ is a sequence in $\C$, and the norm 
\begin{equation*}
\|\mu\|_p\coleqq \begin{cases}
\left(\sum_{\lambda\in\Lambda}\abs{\alpha_\lambda}^p\right)^{1/p}, &\quad\text{if }p\in[1,\infty)\\
\sup_{\lambda\in\Lambda}\abs{\alpha_\lambda}, &\quad\text{if }p=\infty\\
\end{cases}
\end{equation*}
 is finite. For such measures we define $\supp{(\mu)}\coleqq\{\lambda\in \Lambda:\alpha_\lambda\neq 0 \}$. Furthermore, for $s>0$, we let $\msc{M}_s^p=\{\mu\in\msc{M}^p:\sep(\supp(\mu))\geq s\}$.

For a complex number $\lambda=\tau+i\nu$ (or the corresponding point $(\tau,\nu)\in \R^2$), we write $(\m{M}_\nu x )(t) \coleqq e^{2\pi i\nu t}x (t)$ for the modulation operator, $(\m{T}_\tau x )(t) \coleqq x (t - \tau)$ for the translation operator, and $\pi(\lambda)=\modu_\nu\tra_\tau$ for the combined time-frequency shift operator. Recall that, for a nonzero Schwartz test function $\varphi\in\schwartzSpace{\R}$, the short-time Fourier transform (STFT) with respect to the window function $\varphi$ is the map $\m{V}_\varphi$ taking Schwartz distributions on $\R$ to complex-valued functions on $\C$ according to
\begin{equation*}
 (\m{V}_\varphi x)(\lambda)=\langle x, \pi(\lambda) \varphi \rangle_{\m{S}'(\R)\times \m{S}(\R)},\quad \text{for } x\in \m{S}'(\R), \lambda\in \C,
\end{equation*}
where $\m{S}'(\R)$ denotes the set of tempered distributions on $\R$.
 We take $\varphi(t)= 2^{\frac{1}{4}}e^{-\pi t^2}$ to be the $L^2$-normalized gaussian and, following  \cite{Groechenig2000}, we write
\begin{equation*}
M^p_m(\R)=\left\{x\in \m{S}'(\R) : \| x\|_{M^p_m(\R)}  \coleqq \left( \int_{\C} | (\m{V}_\varphi x )(\lambda)|^p m(\lambda)^p  \mathrm{d} \lambda \right)^{1/p}<\infty \right\},
\end{equation*}
for the weighted modulation space on $\R$ of index $p$ and weight function $m:\C\to \R_{\geq 0}$. When $m\equiv 1$, we write $M^p(\R)$ for the unweighted modulation space.
We remark that $\varphi$ has the convenient property of being its own Fourier transform, i.e., $\widehat{\varphi}=\varphi$. According to \cite[Thm. 11.3.5, Thm. 11.3.6]{Groechenig2000}, $M^p(\R)$ is a Banach space, and, for  $p\in[1,\infty)$, its dual space can be identified with $M^q(\R)$ via the dual pairing
\begin{equation}\label{MSDualPairing}
\langle f,g\rangle_{M^p(\R)\times M^q(\R)}=\langle \m{V}_\varphi f,\m{V}_\varphi g\rangle_{L^p(\C)\times L^q(\C)},\quad\text{for }f\in M^p(\R), g\in M^q(\R).
\end{equation}

Finally, for real-valued functions $f$ and $g$ of several variables $p_1,\dots, p_n$ (which may be real or complex numbers, or even functions), and a non-negative integer $m\leq n$, we write $f\lesssim_{\, p_1,\dots,p_m} g$ if there exists a non-negative function $C=C(p_1,\dots, p_m)$ such that $f\leq C\,g$, as well as $f \asymp_{\, p_1,\dots,p_m} g$ if $f \lesssim_{\, p_1,\dots,p_m} g$ and $g \lesssim_{\, p_1,\dots,p_m} f$. We use the notation $f\lesssim g$ only if $C$ is a universal constant, i.e., if it is independent of all of the $p_1,\dots, p_n$.

\section{Contributions}

\subsection{Operators and identifiability}
In order to formalize our definition of identifiability \eqref{eq:ident-proto}, we first need to make sense of the integral \eqref{eq:measure-int-proto}. Concretely, we consider only probing signals $x$ in the modulation space $\MSp{1}{\R} $ (also referred to in the literature as $\mathcal{S}_0$, the Feichtinger algebra) and, for a measure $\mu\in \msc{M}^p$, we interpret  \eqref{eq:measure-int-proto} as a linear operator $\linOp_\mu \colon \MSp{1}{\R} \rightarrow \MSp{p}{\R}$ given by
\begin{equation*}
\begin{aligned}
 \linOp_\mu x  &= \int_{\R^2} x (\cdot - \tau)e^{2\pi i\nu \,\cdot }\,\mrm{d}{\mu}(\tau, \nu)\coleqq \sum_{\lambda \in\mrm{supp}(\mu)}\mu(\{\lambda\})\,\pi(\lambda) x.
\end{aligned}
\end{equation*}
The convergence of this sum in the Banach space $\MSp{p}{\R}$ is guaranteed by the following proposition whose proof can be found in the appendix. 

\begin{prop}\label{wellDefProp}
Let $\Lambda$ be a relatively separated subset of $\C$, and let $p\in[1,\infty)$. Then
\begin{enumerate}[(i)]
\item $\linOp_{\Lambda}:\ell^p(\Lambda)\times M^1(\R)\to M^p(\R)$ given by
\begin{equation*}
 \linOp_\Lambda(\alpha,x )\coleqq  \sum_{\lambda \in\Lambda } \alpha_\lambda \pi(\lambda) x,\quad 	\text{for all }\alpha\in \ell^p(\Lambda), \; x\in M^1(\R) ,
\end{equation*}
 is a well-defined continuous linear operator, in the sense of the sum converging unconditionally in the norm of $M^p(\R)$. Moreover, this operator is bounded according to
\begin{equation*}
\|\linOp_{\Lambda} (\alpha,x)\|_{M^p(\R)}\lesssim \mrm{rel}(\Lambda)\, \|\alpha\|_{\ell^p} \|x\|_{M^1(\R)},
\end{equation*}
for all $\alpha\in \ell^p(\Lambda)$, $x\in M^1(\R)$.
\item For a fixed $x\in M^1(\R)$, the adjoint operator $\big(\linOp_{\Lambda}(\cdot,x)\big)^*: M^q(\R)\to \ell^q$ of the map $\ell^p\ni \alpha\mapsto \linOp_{\Lambda}(\alpha,x)$ is given by
\begin{equation}\label{HmuDual-def}
\big(\linOp_{\Lambda}(\cdot,x)\big)^*(y)=\{\langle y,\pi(\lambda)x\rangle_{M^q(\R)\times M^p(\R)}\}_{\lambda\in\Lambda},\quad \text{for }y\in M^q(\R).
\end{equation}
\end{enumerate}
Next, for a measure $\mu=\sum_{\lambda\in\Lambda}\alpha_\lambda\delta_\lambda\in\msc{M}^p$, define $ \linOp_\mu:M^1(\R)\to M^p(\R)$ by $ \linOp_\mu (x )=  \linOp_{\supp(\mu)} (\alpha,x )$. Then,
\begin{enumerate}[label=(\roman)]
\item[(iii)] for every  $\mu\in\msc{M}^p$,
\begin{equation}\label{eq:wellDefProp-0}
\|\mu\|_{\ell^{\infty}} \lesssim \|\linOp_{\mu} \|_{M^1(\R)\to M^p(\R)}.
\end{equation}
\end{enumerate}
\end{prop}

\noindent As a consequence of item (iii) of Proposition \ref{wellDefProp}, we have
\begin{equation*}
\|\mu_1-\mu_2\|_{\ell^{\infty}} \lesssim \|\linOp_{\mu_1-\mu_2} \|_{M^1(\R)\to M^p(\R)}=\|\linOp_{\mu_1}- \linOp_{\mu_2} \|_{M^1(\R)\to M^p(\R)},
\end{equation*}
and therefore $\mu_1=\mu_2$ whenever $\linOp_{\mu_1}=\linOp_{\mu_2}$. In other words, the measures in $\msc{M}^p$ are completely characterized by their action on $\MSp{1}{\R}$, and thus there is a one-to-one correspondence between the measures in $\msc{M}^p$ and the operators $\{\linOp_\mu:\mu\in\msc{M}^p\}$. Note that this property is necessary for there to be any hope of recovering a measure $\mu$ from a measurement $ \linOp_\mu x $ with respect to a \emph{single} probing signal $x$. 

We are now ready to state our definition of stable identifiability:

\begin{definition}[Stable identifiability]
	\label{identifiability}  \label{defn: identifiability}
	Let $p \in [1, \infty)$. We say that a class of measures $\msc{H}\subset\msc{M}^p$ is stably identifiable by  a probing signal $x\in M^1(\R)$ if there exist constants $C_1,C_2>0$ (that may depend on $p$ and $x$) such that
	\begin{equation}
		C_1\left(\mrm{ms}(\Lambda_1,\Lambda_2)\wedge 1\right)\norm{\mu_1-\mu_2}_{p} \leq \norm{\linOp_{\mu_1}x -\linOp_{\mu_2}x }_{\MSp{p}{\R}} \leq C_2\norm{\mu_1-\mu_2}_{p},
		\label{eq: identifiability set of operators}
	\end{equation}
	 for all $\mu_1,\mu_2\in\msc{H}$, where $\Lambda_j\coleqq \supp(\mu_j)$, $j\in\{1,2\}$.
\end{definition}

\noindent The significance of the term $\mrm{ms}(\Lambda_1,\Lambda_2)$ in \eqref{eq: identifiability set of operators} becomes apparent when we consider classes $\msc{H}$ that contain measures with potentially arbitrarily close supports. For a concrete example, consider the class $\msc{H}=\{\mu\in\msc{M}^p,\# (\mrm{supp}(\mu))=1\}$ of single time-frequency shifts. This class contains the measures $\mu=\delta_{(0,0)}$ and $\mu_\epsilon=\delta_{(0,\epsilon)}$, for all $\epsilon>0$. Let $x\in M^1(\R)$ be a probing signal satisfying the time-localization constraint $t\, x(t)\in L^2(\mathrm{d}t)$, but otherwise arbitrary. Then $\mrm{ms}(\supp(\mu),\supp(\mu_\epsilon))=\epsilon$, and 
\begin{equation*}
\epsilon^{-1}(\linOp_{\mu}x-\linOp_{\mu_\epsilon}x)=\epsilon^{-1}(1-e^{2\pi i \epsilon\,\cdot })\, x\to -(2\pi i\,\cdot)x
\end{equation*}
in $L^2(\R)=M^2(\R)$ as $\epsilon \to 0$, and hence 
\begin{equation}\label{eq:Mms-ratio}
\|\linOp_{\mu}x-\linOp_{\mu_\epsilon}x\|_{M^2(\R)}/ \mrm{ms}(\supp(\mu),\supp(\mu_\epsilon)) \asymp \|t \, x(t)\|_{L^2(\mathrm{d}t )} >0,
\end{equation}
 as $\epsilon\to 0$. On the other hand, $\|\mu-\mu_{\epsilon}\|_2=\sqrt{2}$ is bounded away from $0$ as $\epsilon\to 0$. Thus, if the class $\msc{H}$ is to be identifiable, the lower bound in \eqref{eq:ident-proto} needs to decay at least linearly with $\mrm{ms}(\supp(\mu_1),\supp(\mu_2))$, for $\mu_1,\mu_2\in \msc{H}$. 
 In contrast to \eqref{eq:Mms-ratio}, one could have another class $\msc{K}$ containing measures $\mu'$ and $\mu_\epsilon'$, for $\epsilon>0$, so that $\|\linOp_{\mu'}x-\linOp_{\mu'_\epsilon}x\|_{M^2(\R)}/\mrm{ms}(\supp(\mu'),\supp(\mu'_\epsilon))\to 0$, as $\epsilon\to 0$, i.e.,  $\|\linOp_{\mu'}x-\linOp_{\mu'_\epsilon}x\|_{M^2(\R)}$ decays superlinearly with $\mrm{ms}(\supp(\mu'),\supp(\mu'_\epsilon))$. Classes such as $\msc{K}$ are not covered by our theory, and we hence exclude them from our definition of stable identifiability. In summary, Definition \ref{defn: identifiability} says that we consider a class of measures to be stably identifiable if the decay of $\|\linOp_{\mu_1}x-\linOp_{\mu_2}x\|_{M^2(\R)}$ as $\mrm{ms}(\supp(\mu_1),\supp(\mu_2))\to 0$ is not faster than linear. This property will turn out to be crucial later when we discuss robust recovery (specifically, in the proofs of Theorems \ref{w*w*thm} and \ref{RobustConvThm}).

 \subsection{A necessary and sufficient condition for identifiability}
 
 As already mentioned in the introduction, our necessary and sufficient condition for identifiability will be expressed in terms of the density of support sets measured uniformly over the class of measures under consideration. Concretely, we have the following definition:
 \begin{definition}[Upper Beurling class density]\label{def:UBcd}
	Let $\m{L}$ be a collection of relatively separated sets in $\R^2$, and, for $R>0$, define $[0,R]^2=[0,R]\times[0,R]\subset  \R^2$. For $\Lambda\in\m{L}$, let $n^+(\Lambda, [0,R]^2)$ be the largest number of points of $\Lambda$ contained in any translate of $[0,R]^2$ in the plane. We then define the upper Beurling class density of $\m{L}$ according to
	\begin{equation*}
		\m{D}^+(\m{L})= \limsup_{R \to \infty} \sup_{\Lambda\in \m{L}} \frac{n^+(\Lambda, [0,R]^2)}{R^2}.
	\end{equation*}
\end{definition}
 
  We are now ready to state the first main result of the paper.

 \begin{thm}[A sufficient condition for identifiability]\label{CritThm}
Let $p\in(1,\infty)$ and $s>0$, let $\msc{H}\subset\msc{M}_s^p$ be a class of measures, and set $\m{L}=\{\supp(\mu) : \mu \in \msc{H}\}$. Suppose that $ \m{D}^+(\m{L})<\frac{1}{2}\,$. Then the class $\msc{H}$ is identifiable by the standard gaussian $\varphi(t)=2^{\frac{1}{4}}e^{-\pi t^2}$, $\varphi\in \MSp{1}{\R}$.
\end{thm}

\noindent Crucially, the support sets $\supp(\mu)\in\m{L}$ in Theorem \ref{CritThm} are not assumed to be subsets of a lattice or any other a priori fixed subset of $\R^2$. In particular, one allows $\msc{H}$ to contain measures $\mu_1$ and $\mu_2$ with arbitrarily small $\mrm{ms}(\supp(\mu_1),\supp(\mu_2))$.
 
Note that a subclass $\msc{H}'$ of an identifiable class $\msc{H}$ is trivially identifiable, and accordingly the upper Beurling class density of the supports of measures in $\msc{H}'$  does not exceed that of the support sets corresponding to $\msc{H}$. The sufficiency result in Theorem \ref{CritThm} is therefore ``compatible'' with the inclusion relation on classes.
By contrast, the ``non-identifiability'' of a class $\msc{H}\subset\msc{M}^p_s$ (i.e., the nonexistence of a probing signal in $M^1(\R)$ by which the class would be identifiable) can only be meaningfully assessed in terms of the Beurling density $\m{D}^+(\{\supp(\mu):\mu\in\msc{H}\})$ for sufficiently rich classes of measures. For example, one can construct arbitrarily large finite subsets $\msc{H}$ of $\msc{M}^p_s$ with arbitrarily large $\m{D}^+(\{\supp(\mu):\mu\in\msc{H}\})$, and yet $\msc{H}$ will be identifiable (e.g. by the standard gaussian, using the property that distinct time-frequency shifts of a gaussian are linearly independent). A converse statement to Theorem \ref{CritThm} can hence be meaningfully formulated only for classes $\msc{H}$ that are ``sufficiently rich'' in a suitable sense.
 In the present paper we will do this for classes of measures that are subspaces of $\msc{M}^p$, with support sets that are closed under limits with respect to weak convergence and 
invariant under time-frequency shifts.

Before providing the precise definition of these classes of measures, we need to introduce the notion of weak convergence for subsets of $\C$. Concretely, we say that a sequence of separated subsets $\{\Lambda_n\}_{n\in\N}$ {converges weakly} to $\Lambda\subset\C$, and write $\Lambda_n\xrightarrow[]{w} \Lambda$,  if
\begin{equation}\label{eq:set-weak-conv}
\mrm{dist}\big((\Lambda_n\cap B_R(z))\cup \partial B_R(z), (\Lambda\cap B_R(z))\cup \partial B_R(z)\big)\to 0\quad\text{as }n\to\infty,
\end{equation}
for all $R>0$ and $z\in\C$, where $\mrm{dist}$ denotes the Hausdorff metric on the subsets of $\C$. We are now ready to formalize the type of classes covered by our necessity result.

\begin{definition}[Regular  $\msc{H}(\m{L})^p$ classes]\label{defn: Classes of closed supports}\hfil\\
Let $p\in(1,\infty)$ and $s>0$, and let $\m{L}$ be a collection of separated subsets of $\C$. 
\begin{enumerate}[(i)]
\item We say that $\m{L}$ is \emph{closed and shift-invariant} (CSI) if it is closed under limits with respect to weak convergence, and $\Lambda+z\coleqq \{\lambda+z :\lambda\in\Lambda\}\in\m{L}$, for all $\Lambda\in\m{L}$ and $z\in\C$.

\item  We define a class of measures $\msc{H}(\m{L})^p\subset \msc{M}^p $ according to
\begin{equation*}
\msc{H}(\m{L})^p=\Bigg\{\sum_{\lambda\in\Lambda} \alpha_\lambda \delta_\lambda: \Lambda\in\m{L}, \alpha\in\ell^p(\Lambda) \Bigg\}.
\end{equation*}
We call $\msc{H}(\m{L})^p$ \emph{$s$-regular} if $\m{L}$ is CSI and $\sep(\Lambda)\geq s$, for all $\Lambda\in \m{L}$.
\end{enumerate}
\end{definition}

\noindent Even though the conditions in Definition \ref{defn: Classes of closed supports} are rather technical, they are not overly restrictive, as evidenced by several examples of $s$-regular classes provided in  \S\ref{subsec:examples}. We are now ready to state our second main result, which is a necessary condition for identifiability of $s$-regular classes and as such constitutes a partial converse to Theorem \ref{CritThm}.

\begin{thm}[A necessary condition for identifiability of $s$-regular classes]\label{NeceCritThm}
Let $p\in(1,\infty)$ and $s>0$, and let $\msc{H}(\m{L})^p\subset\msc{M}^p$ be an $s$-regular class. If there exists an $x\in \MSp{1}{\R}$ such that $\msc{H}(\m{L})^p$ is identifiable by $x$, then $\m{D}^+(\m{L})< \frac{1}{2}$.
\end{thm}
 
 \subsection{Identifiability and robust recovery}
 
 In this subsection we formalize the claim \eqref{eq:proto-robust-intro} made in the introduction under the assumption that $x\in M^1_m(\R)$, with the weight function $m(z)=1+|z|$, for $z\in\C$. Informally, this assumption imposes faster-than-linear decay on $x$ in both the time and frequency domains. Note that the $L^2$-normalized gaussian $\varphi$ is in $M^1_m(\R)$, as its STFT decays exponentially (by virtue of $\varphi\in\schwartzSpace{\R}$ and \cite[Thm. 11.2.5]{Groechenig2000}).
 
We begin by defining the weak-* topology on $\msc{M}_s^p$, for $p\in (1,\infty)$. Concretely, for $\mu\in\msc{M}_s^p$ and a sequence $\{\mu_n\}_{n\in\N} \subset \msc{M}_s^p$, we say that $\{\mu_n\}_{n\in\N}$ converges to $\mu$ in the weak-* topology of $\msc{M}_s^p$, and write $\mu_{n}\xrightarrow[]{w^*} \mu $, if
\begin{equation}\label{eq:M-w*-conv}
\lim_{n\to\infty} \int_\C \overline{f} \mrm{d}\mu_n = \int_\C \overline{f} \mrm{d}\mu,
\end{equation}
for all continuous $f:\C \to\C$ such that $\lim_{|z|\to\infty}f(z)= 0$ and
\begin{equation*}
\Big\|\sup_{\substack{y\in\C, |y| \leq 1}} |f(z+y)| \Big\|_{L^q(\mrm{d}z)}<\infty.
\end{equation*}
This definition corresponds to convergence in the weak-* topology on the Wiener amalgam space $W(\m{M},L^p)$, which will be defined and  treated  systematically in \S \ref{sec:LatticesWiener}. In order to formalize \eqref{eq:proto-robust-intro}, it will be helpful to first state the following weak-* recovery result for $s$-regular classes:

\begin{thm}[Weak-* Recovery Theorem]\label{w*w*thm}
Let $p\in(1,\infty)$ and $s>0$, and let $\msc{H}(\m{L})^p \subset \msc{M}_s^p$ be an $s$-regular class. Assume furthermore that $\msc{H}(\m{L})^p$ is identifiable by a probing signal $x\in M^1_m(\R)$, where $m(z)=1+|z|$. Then
\begin{itemize}
\item[(i)] if $\mu,\wtd{\mu}\in\msc{H}(\m{L})^p$ are such that $\linOp_{\wtd\mu}\,x= \linOp_{\mu}x$, then $\wtd{\mu}=\mu$.
\item[(ii)] Let $\mu\in\msc{H}(\m{L})^p $ and let $\{\mu_n\}_{n\in\N}$ be a sequence in $\msc{H}(\m{L})^p$. Then $\linOp_{\mu_n}x\to \linOp_{\mu}x$ in the weak-* topology of $M^p(\R)$ if and only if $\mu_{n}\xrightarrow[]{w^*} \mu$ in $\msc{M}_s^{p}$.
\end{itemize}
\end{thm}
\noindent The proof of Theorem \ref{w*w*thm} relies crucially on the fact that the decay of the lower bound in \eqref{eq: identifiability set of operators} as a function of $\mrm{ms}(\supp(\mu_1),\supp(\mu_2))$ is not faster than linear. 

Note that item (i) of Theorem \ref{w*w*thm} guarantees perfect recovery of measures in $\msc{H}(\m{L})^p$ under perfect measurement matching. However, this does not go a long way towards establishing \eqref{eq:proto-robust-intro} as item (ii) of the theorem deals with convergence in weak-* topologies ``only''. To illustrate that a stronger form of convergence is needed, consider the $1/2$-regular class $\msc{H}(\m{L})^2$, where $\m{L}=\{\Lambda\subset \C : \sep(\Lambda)\geq 1/2, \#(\Lambda)\leq 2 \}$. In this class $\delta_{0,0}+\delta_{n,0}\xrightarrow[]{w^*} \delta_{0,0}$ as $n\to\infty$ (where $\delta_{\tau,\nu}$ is again the Dirac point measure with mass at $(\tau,\nu)$), and so, if one were to rely on the weak-* convergence guarantee only, one could argue that $\{\delta_{0,0}+\delta_{n,0}\}_{n\in\N}$ recovers $\delta_{0,0}$. This sequence does, indeed, capture the component $\delta_{0,0}$, but it also features the nonvanishing spurious component  $\delta_{n,0}$. Similarly, on the measurement side of \eqref{eq:proto-robust-intro}, taking $x=\varphi$ as the probing signal would yield $\varphi+\varphi(\,\cdot-n)\to\varphi$ in the weak-* topology of $L^2$, but not in the norm topology. We can thus hope that upgrading from weak-* convergence to norm convergence on the measurement side of \eqref{eq:proto-robust-intro}  might imply a stronger form of convergence of the sequence of candidate measures to the target measure. The following theorem establishes that this is, indeed, the case for $s$-regular classes $\msc{H}(\m{L})^p$. Concretely, convergence of the measurements in norm implies that the candidate measures $\mu_n$ approximate arbitrarily big finite sections of the target measure $\mu$ and do not have any spurious components.

\begin{thm}[Robust Recovery Theorem]\label{RobustConvThm}
 Let $p\in(1,\infty)$ and $s>0$, and let $\msc{H}(\m{L})^p \subset \msc{M}_s^p$ be an $s$-regular class. Assume furthermore that $\msc{H}(\m{L})^p$ is identifiable by a probing signal $x\in M^1_m(\R)$, where $m(z)=1+|z|$, and let $C_1$ and $C_2$ be the corresponding constants such that \eqref{eq: identifiability set of operators} is fulfilled. Fix a $\mu\in\msc{H}(\m{L})^p$, write $\Lambda=\supp(\mu)$, and let $\{\mu_n\}_{n\in\N}$ be a sequence in $\msc{H}(\m{L})^p$ such that $\|\linOp_{\mu_n}x-\linOp_{\mu}x\|_{M^p(\R)}\to 0$ as $n\to\infty$.

Then, for every $\epsilon>0$ and every finite subset $\wtd{\Lambda}$ of $\Lambda$ such that $\|\mu-\mu\mathds{1}_{\wtd{\Lambda}}\|_{p}<\epsilon$, there is an $N\in\N$ so that, for all $n\geq N$, the measures $\mu_n$ take the form 
\begin{equation*}
\mu_n=\sum_{{\lambda}\in\tilde{\Lambda}}\alpha^{(n)}_{\lambda}\delta_{\lambda+\bve_n(\lambda)}+\rho_n,
\end{equation*}
where $|\bve_n(\lambda)|\leq \epsilon$ and $|\alpha^{(n)}_{\lambda}-\alpha_\lambda|\leq \epsilon$, for all $\lambda\in\wtd{\Lambda}$, and $\|\rho_n\|_p\leq\frac{4C_2}{C_1(s\wedge 1)}\,\epsilon\,$.
\end{thm}

 One can view $\sum_{{\lambda}\in\tilde{\Lambda}}\alpha^{(n)}_{\lambda}\delta_{\lambda+\bve_n(\lambda)}$ as the ``successfully recovered finite section'' of $\mu$, which approximates both the time-frequency shifts and their weights within $\epsilon$ error, whereas $\rho_n$ is the ``spurious'' component, whose norm is also proportional to $\epsilon$. The constant of proportionality $(C_2/C_1)\cdot (s\wedge 1)^{-1}$ in the bound on $\|\rho_n\|_p$ can be interpreted as a ``condition number'', indicating that the spurious component is more difficult to suppress when the ratio of identifiability constants $C_2/C_1$ is large, or when the separation $s$ of the measures under consideration is excessively small, which agrees with our intuition on the behavior of the ``difficult cases''.
 
 \subsection{Examples of identifiable and non-identifiable \texorpdfstring{$s$}{s}-regular classes\label{subsec:examples}}
Finally, we present several explicit families of $s$-regular classes and discuss their identifiability in view of Theorems \ref{CritThm} and \ref{NeceCritThm}. 
Let $p\in(1,\infty)$, $s>0$, $N\in\N$, $\theta >0$, and $R>0$, and define the sets 
\begin{equation*}
\begin{aligned}
\m{L}_{s}^{\text{sep}}&=\{\Lambda\subset\R^2 :\sep(\Lambda)\geq s\},\\
\m{L}_{s,N}^{\text{fin}}&=\{\Lambda\subset\R^2 :\sep(\Lambda)\geq s, \#(\Lambda)\leq N\},\text{ and}\\
\m{L}_{s,\theta,R}^{\text{Ray}}&=\{\Lambda\subset\R^2 :\sep(\Lambda)\geq s, \, n^+\! \left(\Lambda,(0,R)^2\right)\leq \theta R^2\}.
\end{aligned}
\end{equation*}
We call the corresponding sets $\msc{H}({\m{L}_{s}^{\text{sep}}})^p$, $\msc{H}({\m{L}_{s,N}^{\text{fin}}})^p$, and $\msc{H}({\m{L}_{s,\theta,R}^{\text{Ray}}})^p$, the $\ell^2$-separated, finite, and Rayleigh classes, respectively. The following proposition shows that these classes are $s$-regular.

\begin{prop}\label{prop:examples-CSI}
Let $s>0$, $N\in\N$, $\theta>0$, and $R>0$. Then the collections $\m{L}_{s}^{\text{sep}}$, $\m{L}_{s,N}^{\text{fin}}$, and $\m{L}_{s,\theta,R}^{\text{Ray}}$ are CSI and so the corresponding  $\ell^2$-separated, finite, and Rayleigh classes are $s$-regular.
\end{prop}

\noindent Theorems \ref{CritThm} and \ref{NeceCritThm} can be used to obtain the following identifiability results for these classes.

\begin{cor}[Finite class]\label{cor:fin-class}
Let $p\in(1,\infty)$, $s>0$, and $N\in\N$. Then the class $H(\m{L}_{s,N}^{\text{fin}})^p$ is identifiable by the gaussian $\varphi(t)=2^{\frac{1}{4}} e^{-\pi t^2}$.
\end{cor}

\begin{cor}[$\ell^2$-separated class]\label{cor:SuffSepClass}
Let $p\in(1,\infty)$ and $s>0$. Then,
\begin{itemize}
\item[(i)] if $s>2\cdot 3^{-\frac{1}{4}}$, $\msc{H}({\m{L}_{s}^{\text{sep}}})^p$ is stably identifiable by $\varphi$, and
\item[(ii)] if $s\leq 2\cdot 3^{-\frac{1}{4}}$, $\msc{H}({\m{L}_{s}^{\text{sep}}})^p$ is not stably identifiable by any probing signal.
\end{itemize}
\end{cor}

\begin{cor}[Rayleigh class]\label{cor:RayClass}
Let $p\in(1,\infty)$ and $s\in(0,\theta^{-1/2})$. Then,
\begin{itemize}
\item[(i)] if $\theta<\frac{1}{2}$, $\msc{H}({\m{L}_{s,\theta,R}^{\text{Ray}}})^p$ is stably identifiable by $\varphi$, for all $R>0$, and
\item[(ii)] if $\theta> \frac{1}{2}$, there exists an $R_0>0$ such that $\msc{H}({\m{L}_{s,\theta,R}^{\text{Ray}}})^p$ is not stably identifiable by any probing signal, for all $R\geq R_0$.
\end{itemize}
\end{cor}

One could also consider the class $\msc{H}(\{\Xi\})^p$ for a fixed lattice $\Xi= A (\Z\times \Z) + b $, where $A\in\R^{2\times 2}$ and $b\in\R^2$, in which case $\msc{H}(\{\Xi\})^p$ is 
$\sep(\Xi)$-separated, stably identifiable by $\varphi$ if $\det(A)> 1$, and not stably identifiable by any probing signal if $\det(A)\leq 1$.


\section{Lattices, Beurling densities, and Wiener amalgam spaces}\label{sec:LatticesWiener}

In this section we introduce various technical tools used throughout the paper. We begin with square lattices in $\C$ and write $\Omega_\gamma=\{\omega_{m,n}=\gamma(m+in)\}_{m,n\in\Z}$ for the square lattice in $\C$ of mesh size $\gamma>0$. Whenever we identify $\C$ with $\R^2$, $\Omega_\gamma$ is equivalently given by $\{(\gamma m, \gamma n):m,n\in\Z\}$. Next, we define the (standard) upper Beurling density, which is analogous to our Definition \ref{def:UBcd}, but is defined for individual subsets of $\R^2$, instead of classes of subsets.

\begin{definition}[Upper Beurling density, {\cite[p.~346]{Beurling1989-1}\cite[p.~47]{Landau1966}}]\label{def:standard-B-dens}
	Let $\Lambda$ be a relatively separated set in $\R^2$, and, for $R>0$, let $[0,R]^2\subset  \R^2$. Let  $n^+(\Lambda, [0,R]^2)$ be the largest number of points of $\Lambda$ contained in any translate of $[0,R]^2$. We then define
	\begin{equation*}
	 D^{+}(\Lambda) \coleqq  \limsup_{R \to \infty} \frac{n^+(\Lambda, [0,R]^2)}{R^2},
	\end{equation*}
	and we call this quantity the upper (standard) Beurling density of $\Lambda$.
\end{definition}

The following three lemmas, whose proofs can be found in the appendix, relate the lattices $\Omega_\gamma$, the upper Beurling class density, and the standard Beurling density.

\begin{lem}\label{lem:Rayleigh-number}
Let $\m{L}$ be a collection of relatively separated sets in $\R^2$, and suppose that $\m{D}^+(\m{L})<\infty$. Then,
\begin{enumerate}[(i)]
\item for every $\theta> \m{D}^+(\m{L})$, there exists an $R_0>0$ such that
\begin{equation*}
n^+(\Lambda , (0,R)^2)\leq \theta R^2,
\end{equation*}
for all $\Lambda\in \m{L}$ and $R\geq R_0$, and
\item $\m{D}^+(\m{L})\geq \sup_{\Lambda \in \m{L}} D^+(\Lambda)$.
\end{enumerate}
\end{lem}

\begin{definition}\label{def:unif-close-to-lat}
Let $\Lambda$ be a non-empty relatively separated subset of $\C$, and let $\gamma>0$ and $R>0$. We say that $\Lambda$ is $R$-uniformly close to $\Omega_\gamma=\{\omega_{m,n}=\gamma(m+in)\}_{m,n\in\Z}\,$ if there exists an enumeration $\{\lambda_{m,n}\}_{(m,n)\in\m{I}}$ of $\Lambda$ (with index set $\m{I}\subset\Z\times\Z$) such that $|\lambda_{m,n}-\omega_{m,n}|\leq R$, for all $(m,n)\in\m{I}$. 
\end{definition}

\begin{lem}\label{lem:unif-lat-dist}
Let $\Lambda$ be a non-empty discrete set in $\C$, and let $\theta>0$, $\gamma>0$, and $R>0$. If $\gamma^{-2}>\theta$ and $n^+(\Lambda , (0,R)^2)\leq \theta R^2$, then there exists an $R'=R'(\theta,\gamma,R)>0$ such that $\Lambda$ is $R'$-uniformly close to $\Omega_\gamma$.
\end{lem}

\begin{lem}\label{lem:class-lat-dist}
Let $\m{L}$ be a set of relatively separated subsets of $\C$, and let $\gamma>0$ and $R>0$. If $\Lambda$ is $R$-uniformly close to $\Omega_\gamma$, for all $\Lambda\in \m{L}$,  then $\m{D}^+(\m{L})\leq \gamma^{-2}$.
\end{lem}

We conclude this section by formalizing Wiener amalgam spaces \cite{Feichtinger1980,Feichtinger1983} on $\C$ and relating them to weak-* convergence on $\msc{M}_s^p$ defined in \eqref{eq:M-w*-conv}. We adopt most of our terminology from \cite{Groechenig2015}. Let $\m{D}(\C)$ be the test space of smooth compactly supported functions on $\C$, with its usual inductive limit topology and the corresponding topological dual $\m{D}'(\C)$, called the space of distributions. Let $\mathcal{B}$ be a Banach space that admits a continuous embedding into $\m{D}'(\C)$. Furthermore, fix a non-negative compactly supported continuous function $\psi\in\m{D}(\C)$ forming a partition of unity, i.e., $\sum_{z\in\Z^2}\psi(\,\cdot-z)=1$, and let $m:\C\to\R_{\geq 0}$ be a weight function of the form $m(z)=(1+|z|)^r$, for some $r\geq 0$. Then, for $p\in[1,\infty]$, the Wiener amalgam space $W(\mathcal{B},L^p_m)$ is defined as
\begin{equation*}
W(\mathcal{B},L^p_m)=\left\{f\in \m{D}'(\C):\|f\|_{W(\mathcal{B},L^p_m)}\coleqq\Big\| \| f\, \overline{\psi}(\,\cdot -z)\|_\mathcal{B} \, m(z) \Big\|_{L^p(\mathrm{d}z )}<\infty \right\}.
\end{equation*}
The definition of $W(\mathcal{B},L^p_m)$ is independent of the choice of $\psi$, and different $\psi$ define equivalent norms on $W(\mathcal{B},L^p_m)$. Informally, $W(\mathcal{B},L^p_m)$ is the space of distributions (i.e., generalized functions) on $\C$ that are ``locally in $\mathcal{B}$'' and ``globally in $L^p_m$''.

Next, we claim that $\msc{M}_s^p\subset W(\m{M},L^p)$, for $p\in (1,\infty]$, where $\m{M}$ is the space of regular complex-valued Borel measures on $\C$ with the total variation norm. To see this, let $r>0$ be such that $\supp(\psi)\subset B_r(0)$. Now, for a measure $\mu=\sum_{\lambda\in\Lambda} \alpha_\lambda\delta_\lambda \in \msc{M}_s^p$ denote $|\mu|^p\coleqq \sum_{\lambda\in\Lambda} |\alpha_\lambda|^p\,  \delta_\lambda$. Then Hölder's inequality yields
\begin{equation*}
\begin{aligned}
\| \mu\, \overline{\psi}(\,\cdot -z)\|_{\m{M}} &\leq \sum_{\lambda\in \Lambda\cap B_r(z)} |\alpha_\lambda|\leq  \Bigg(  \sum_{\lambda\in \Lambda\cap B_r(z)} 1^q \Bigg)^{1/q }  \Bigg(\sum_{\lambda\in \Lambda \cap B_r(z)} |\alpha_\lambda|^p\Bigg)^{1/p}\\
&\lesssim_{\, \psi} \left(\sep(\Lambda)^{-2}\right)^{1/q}\,   \Bigg(\int_{\C} \mathds{1}_{\{ |y - z|\leq r \} }\,  \mrm{d}|\mu|^p(y)  \Bigg)^{1/p}, \quad \text{for } z\in \C,
\end{aligned}
\end{equation*}
where the last inequality follows since one can pack at most $r^2/(\sep(\Lambda)/2)^{-2}$ spheres of radius $\sep(\Lambda)/2$ in $B_r(z)$.
Therefore, as $\sep(\Lambda)\geq s$ Tonelli's theorem yields
\begin{equation*}
\begin{aligned}
\Big\| \| \mu\,\overline{\psi}(\,\cdot -z) \|_{\m{M}}   \Big\|_{L^p(\mathrm{d}z )}^p&\lesssim_{\,\psi,s} \int_{\C}\Bigg[ \int_{\C} \mathds{1}_{\{ |z-y|\leq r \} } \, \mrm{d}|\mu|^p \Bigg]\mrm{d}z\\
&=  \int_{\C}  \underbrace{ \int_{\C}  \mathds{1}_{\{ |z-y|\leq r \} }  \mrm{d}z }_{=\pi r^2}  \,  \mrm{d}|\mu|^p = \pi r^2 \|\mu\|_p^p <\infty,
\end{aligned}
\end{equation*}
and so $\mu\in  W(\m{M},L^p)$. As $\mu \in \msc{M}_s^p$ was arbitrary, we have therefore shown that
\begin{equation}\label{eq:WnormEquiv-0}
\|\mu\|_{W(\m{M},L^p)}\lesssim_{\, \psi, p,s} \|\mu\|_{p},\quad \text{for all }\mu\in\msc{M}_s^p,
\end{equation}
which establishes $\msc{M}_s^p\subset W(\m{M},L^p)$.

Now, by the Riesz-Markov-Kakutani representation theorem \cite[Thm. 6.19]{Rudin1987}, $\m{M}$ can be identified with the topological dual $C_0^*$ of 
\begin{equation*}
C_0=\{f\in L^\infty(\C): f\text{ continuous},\; \lim_{|z|\to\infty}|f(z)|=0\},
\end{equation*}
via the pairing $\langle \mu, f\rangle=\int_\C \overline{f}\mrm{d}\mu$. Therefore, by \cite[Thm. 2.8]{Feichtinger1983}, we have that $|f(y)|\mathds{1}_{B_r(z)}(y)$ is integrable w.r.t. the product measure $ \mrm{d}|\mu|(y) \times \mrm{d}z $ on $\C\times \C$, for $\mu\in W(\m{M},L^p)$ and $f\in W(C_0,L^q)$, and $W(\m{M},L^p)$ can be identified with the topological dual of $W(C_0,L^q)$ via the dual pairing
\begin{equation*}
\langle \mu,f\rangle\coleqq \int_{\C}\left[ \int_\C \overline{f}\, \mathds{1}_{B_r(z)} \mrm{d}\mu\right] \mrm{d}z.
\end{equation*}
An application of Fubini's theorem hence yields
\begin{equation*}
\begin{aligned}
\langle \mu,f\rangle&=\int_{\C}\left[ \int_\C \overline{f}(y)\, \mathds{1}_{B_r(z)}(y) \mrm{d}\mu(y)\right] \;\mrm{d}z= \int_\C \overline{f}(y) \int_{\C} \mathds{1}_{\{|z-y|\leq r\}}\,\mrm{d}z\; \mrm{d}\mu(y)= \pi r^2 \int\overline{f} \mrm{d}\mu.
\end{aligned}
\end{equation*}
Thus, as $\pi r^2$ is a constant depending only on the choice of $\psi$ through $\supp(\psi)\subset B_r(0)$, one can instead use the following simpler dual pairing to effect the correspondence between $W(\m{M},L^p)$ and $W(C_0,L^q)^*$:
\begin{equation*}
\langle \mu,f\rangle= \int\overline{f} \mrm{d}\mu= \sum_{\lambda\in\supp(\mu)}\overline{f(\lambda)}\mu(\{\lambda\}),
\end{equation*}
for $\mu\in W(\m{M},L^p)$ and $f\in W(C_0,L^q)$. Therefore, definition \eqref{eq:M-w*-conv} of weak-* convergence in $\msc{M}_s^p$ corresponds precisely to convergence in the weak-* topology on $W(\m{M},L^p)$ (i.e., the weak topology generated by $W(C_0,L^q)$).

Finally, in the special case of weak convergence of subsets of $\R^2$ defined in \eqref{eq:set-weak-conv}, following \cite[p. 398]{Groechenig2015}, we have that, if $\inf_{n\in\N}\sep(\Lambda_n)>0$, then weak convergence of subsets $\Lambda_n\xrightarrow[]{w} \Lambda$ is equivalent to $\sum_{\lambda\in\Lambda_n}\delta_\lambda\to\sum_{\lambda\in\Lambda}\delta_\lambda$ in the weak-* topology $W(C_0,L^1)$.


\section{Proof of Theorem  \ref{CritThm} }

As already mentioned in the introduction, the proof of  Theorem  \ref{CritThm}  relies on the theory of interpolation of entire functions. The idea for the proof is based on \cite[Thm. 1]{Aubel2015},  where the lower bound (analogous to the left-hand side of \eqref{eq: identifiability set of operators}), however, depends in a non-explicit manner on the supports of the individual measures in the identifiability condition. As our goal is to obtain an explicit lower bound, namely, a constant multiple of the minimum separation of the supports, our theorem needs to be stated in terms of the class density (according to Definition \ref{def:UBcd}) instead of simply considering the standard Beurling density (according to Definition \ref{def:standard-B-dens}) of the supports of the individual measures in the class. This difference will also require us to delve deeper into the interpolation theory underlying the proof of \cite[Thm. 1]{Aubel2015}.

We begin our exposition of the required technical tools by defining the Weierstrass $\sigma_\gamma$-function associated with $\Omega_\gamma=\{\omega_{m,n}=\gamma(m+in)\}_{m,n\in\Z}\,$:
\begin{equation*}
\sigma_\gamma(z)=z\hspace*{-1em}\prod_{(m,n)\in\Z^2\setminus\{(0,0)\}}\left(1-\frac{z}{\omega_{m,n}}\right)\exp{\left(\frac{z}{\omega_{m,n}}+\frac{1}{2}\frac{z^2}{\omega_{m,n}^2}\right)},\quad z\in\C.
\end{equation*}
We will need several basic facts about this function, which can be found in  \cite{Kehe2012} along with a more detailed account of its properties. Concretely, we note that the infinite product in the definition of $\sigma_\gamma$ converges absolutely uniformly on compact subsets of $\C$, and therefore defines an entire function. Moreover, $\sigma_\gamma$ satisfies the following growth estimate: 

\begin{lem}[\!\!{\cite[Cor. 1.21]{Kehe2012}}] \label{vanillaSigmaBd} We have $ |\sigma_\gamma(z)|e^{-\frac{\pi}{2}\gamma^{-2}|z|^2}\asymp_{\,\gamma} d(z,\Omega_\gamma)$, where $d(z,\Omega_\gamma)=\min\{|z-\omega|: \omega\in\Omega_\gamma\}$ denotes the Euclidean distance from $z$ to the lattice $\Omega_\gamma$.
\end{lem}

\noindent 
In order to enable working with measures $\mu$ whose supports are not subsets of lattices, we will need to perturb the zeros of the Weierstrass $\sigma_\gamma$-function. We will do so following \cite{Kehe2012} and \cite[p. 109]{Seip1992-2}. Concretely, let $\m{I}\subset \Z\times\Z$ be an index set with $(0,0)\in\m{I}$, and let $\Lambda=\{\lambda_{m,n}\}_{(m,n)\in\m{I}}$ be a discrete subset of $\C$ with $\lambda_{m,n}\neq 0$ for $(m,n)\in \m{I}\setminus\{(0,0)\}$. We now define the modified Weierstrass function associated with $\Lambda$ by
\begin{equation}\label{ModifiedSigmaDef}
g_\Lambda(z)=(z-\lambda_{0,0})\, \hspace*{-1em}\prod_{(m,n)\in\m{I}\setminus\{(0,0)\}}\left(1-\frac{z}{\lambda_{m,n}}\right)\exp{\left(\frac{z}{\lambda_{m,n}}+\frac{1}{2}\frac{z^2}{\omega_{m,n}^2}\right)},\quad z\in\C.
\end{equation}
According to \cite[Lem. 4.21]{Kehe2012}, provided there exist $\gamma>0$ and $R>0$ such that $\Lambda$ is $R$-uniformly close to $\Omega_\gamma$, expression \eqref{ModifiedSigmaDef} converges uniformly on compact subsets of $\C$ to an entire function with zero set $\Lambda$.
The proof of Theorem  \ref{CritThm} relies on constructing and controlling the growth of an entire function interpolating a sequence of values $\{\beta_\lambda\}_{\lambda\in\Lambda}$ at the points of $\Lambda=\supp(\mu_1)\cup\supp(\mu_2)$, where $\mu_1,\mu_2\in \msc{M}_s^p$ are the measures for which \eqref{eq: identifiability set of operators} is to be established. This will be accomplished by means of ``basis functions'' that interpolate the one-hot sequences $\{\mathds{1}_{\{\lambda=\lambda'\} }\}_{\lambda\in\Lambda}$, for $\lambda'\in\Lambda$. The following lemma furnishes a prototype for these basis functions, obtained by ``dividing out'' a zero of the modified Weierstrass function associated with $\Lambda$, as well as a growth bound reminiscent of \cite{Kehe2012} and \cite[\S 2.2]{Seip1992-2}, with the crucial difference that our bound makes the dependence on the mutual separation of $\supp(\mu_1)$ and $\supp(\mu_2)$ explicit. The proof of the lemma largely follows \cite{Kehe2012}, the only difference being that we need to take the specific form $\Lambda=\supp(\mu_1)\cup\supp(\mu_2)$ of $\Lambda$ into account, carrying out the calculations more explicitly to extract the dependence on the mutual separation of $\supp(\mu_1)$ and $\supp(\mu_2)$.

\begin{lem}\label{massiveLemma}
Let $\Lambda=\{\lambda_{m,n}\}_{(m,n)\in\m{I}}$ be a relatively separated subset of $\C$ with $\lambda_{0,0}=0$. Furthermore, let $\rho$, $s$, $\theta$, $\gamma$, and $R$ be positive real numbers, and set $\Omega_\gamma=\{\omega_{m,n}=\gamma(m+in)\}_{m,n\in\Z}\,$. Define $\m{I}_s=\{(m,n)\in\m{I}:|\lambda_{m,n}|\leq \frac{s}{2}\}$ and suppose that
\begin{enumerate}[(i)]
\item $\#(\m{I}_s)\leq 2$, and, if $\m{I}_s=\{(0,0), (m',n')\}$, then $|\lambda_{m',n'}|\geq \rho\,$,
\item $n^+(\Lambda,(0,R')^2)\leq \theta R'\,^2$, for all $R'\geq R$, and
\item $|\lambda_{m,n}-\omega_{m,n}|\leq R$, for all $(m,n)\in\m{I}$.
\end{enumerate}
Now, let $g_\Lambda$ be given by \eqref{ModifiedSigmaDef} and define $\wtd{g}_{\Lambda}:\C\to\C$ according to 
\begin{equation}\label{eq:massive-def-gtilde}
\wtd{g}_\Lambda(z)=\frac{g_\Lambda(z)}{z}\prod_{(m,n)\in \m{I}_s} \hspace*{-2mm} \exp\left(\frac{z}{\omega_{m,n}}-\frac{z}{\lambda_{m,n}}\right)\prod_{(m,n)\in\Z^2\setminus\m{I}}\left(1-\frac{z}{\omega_{m,n}}\right)\exp{\left(\frac{z}{\omega_{m,n}}+\frac{1}{2}\frac{z^2}{\omega_{m,n}^2}\right)} .
\end{equation}
Then
\begin{itemize}
\item[(a)] $\wtd{g}_\Lambda(0)=1$ and $\wtd{g}_\Lambda(\lambda_{m,n})=0$, for $(m,n) \in \m{I}\setminus\{(0,0)\}$, and
\item[(b)] there exist constants $C>0$ and $c>0$ depending only on $s$, $\theta$, $\gamma$, and $R$ such that
\begin{equation}\label{longLem-statement}
|\wtd{g}_\Lambda(z)|e^{-\frac{\pi}{2}\gamma^{-2}|z|^2}\leq C (\rho\wedge 1)^{-1} 
e^{c|z|\log{|z|}},\quad\text{for all }z\in \C.
\end{equation}
\end{itemize}
\end{lem}
\noindent The proof of Lemma \ref{massiveLemma} can be found in the appendix.

The next preparatory step towards the proof of Theorem \ref{CritThm} is to relate Gabor systems generated by $\varphi(t)=2^{\frac{1}{4}}e^{-\pi t^2}$ with entire functions of suitably bounded growth by means of the Bargmann transform. Concretely, we will work with a definition of the Bargmann transform consistent with \cite{GroWal2001} in order to facilitate arguments involving the isometry property between modulation spaces and Bargmann-Fock spaces introduced next. 
For conjugate indices $p,q \in [1, \infty]$, the Bargmann-Fock space $\BFSp{p}{\C}$ is defined 
 as the set of all entire functions $F$ for which $\norm{F}_{\BFSp{p}{\C}}<\infty$, where
\begin{equation*}
	\norm{F}_{\BFSp{p}{\C}} \coleqq \left(\int_\C \abs{F\left(z\right)}^p e^{-p\pi \abs{z}^2/2} \mrm{d}z\right)^{1/p}, \quad\text{for }p\in[1,\infty),
\end{equation*}
and
\begin{equation*}
\norm{F}_{\BFSp{\infty}{\C}}\coleqq \sup_{z\in\C}{\abs{F(z)}e^{-\pi \abs{z}^2/2}}.
\end{equation*}
The Bargmann transform is now defined as the linear map $\bargmann \colon \MSp{p}{\R}\to\BFSp{p}{\C}$  given by
\begin{equation*}
(\bargmann f)(z)=2^{\frac{1}{4}}e^{-\pi z^2/2}\int_{\R}e^{2\pi tz -\pi t^2}f(t)\mrm{d}t, \quad z\in\C.
\end{equation*}
According to \cite[\S 1.4]{GroWal2001}, the Bargmann transform is an isometric isomorphism between the Banach spaces $\MSp{p}{\R}$  and $\BFSp{p}{\C}$, i.e., it is bijective and 
\begin{equation}\label{eq:Btrans-isom}
\norm{\bargmann f}_{\BFSp{p}{\C}} =\|f\|_{\MSp{p}{\R}}, \quad\text{for all }f\in \MSp{p}{\R}.
\end{equation}
Following \cite{GrycKemp2011}, when $p\in [1,\infty)$, the topological dual of $\BFSp{p}{\C}$ can be identified with $\BFSp{q}{\C}$ via the pairing
\begin{equation*}
\innerProd{F}{G}_{\BFSp{p}{\C}\times\BFSp{q}{\C}}=\int_{\C}F(z)\overline{G(z)}e^{-\pi\abs{z}^2}\mrm{d}z,
\end{equation*}
 for $F\in \BFSp{p}{\C}$ and $G\in\BFSp{q}{\C}$.

The following lemma is a generalization (from $L^2(\R)$ to $M^q(\R)$) of the standard identity \cite[Prop. 3.4.1]{Groechenig2000} relating the Bargmann transform with time-frequency shifts of the gaussian $\varphi(t)=2^{\frac{1}{4}} e^{-\pi t^2}$.

\begin{lem}\label{STFT=Barg-lem}
Let $q\in[1,\infty)$. Then, for every $y\in M^q(\R)$ and $\lambda=\tau+i\nu\in\C$, we have
\begin{equation*}
\langle y,\pi(\lambda)\varphi \rangle_{M^q(\R)\times M^p(\R)}=e^{-\pi i \tau\nu}e^{-\pi |\lambda|^2/2}(\bargmann y)(\bar{\lambda}).
\end{equation*}
\end{lem}

Before finally embarking on the proof of Theorem  \ref{CritThm}, we state the following two lemmas about abstract Banach spaces and Wiener amalgam spaces that will facilitate the application of the more specialized theory of Bargmann-Fock spaces. Their proofs can be found in the appendix.

\begin{lem}\label{ShortTheoryLemma}
Let $\m{A}:X\to Y$ be a continuous linear operator between Banach spaces $X$ and $Y$. We then have the following:
\begin{itemize}
\item[(i)] If $\m{A}$ is bounded below (i.e. there exists a $c>0$ such that $\|\mathcal{A}x\|\geq c\|x\|$, for all $x\in X$), then the adjoint $\mathcal{A}^*:Y^*\to X^*$ is surjective.
\item[(ii)] Suppose that there exists a constant $a>0$ such that, for every $f\in X^*$, there is a $g\in Y^*$ with $\m{A}^* g=f$ and $a\|g\|_{Y^*}\leq \|f\|_{X^*}$. Then $\m{A}$ is bounded from below by $a$.
\end{itemize}
\end{lem}

\begin{lem}\label{IrregConvLem}
Let $p\in[1,\infty)$ and $\Lambda\subset \R^2$ a separated subset and set $s=\sep(\Lambda)>0$. Then
\begin{equation*}
 \Big\| \sum_{\lambda\in\Lambda} \alpha_\lambda \,f(\,\cdot-\lambda) \Big\|_{L^p(\R^2)}\lesssim_{\,p,s} \|f\|_{W(L^\infty,L^1)}\|\alpha\|_{\ell^p(\Lambda)},
\end{equation*}
for all $\{\alpha_\lambda\}_{\lambda\in \Lambda}\subset \ell^p (\Lambda)$ and $f\in W(L^\infty,L^1)$. 
\end{lem}

\begin{proof}[Proof of Theorem \ref{CritThm}]
Fix $\mu_1,\mu_2\in\msc{H}$ and let $\Lambda_j=\supp(\mu_j)$, for $j\in\{1,2\}$, and $\Lambda=\Lambda_1\cup \Lambda_2$. We can then write $\mu_1-\mu_2=\sum_{\lambda\in\Lambda}\alpha_\lambda \delta_\lambda$, where $\alpha\in\ell^p(\Lambda)$, so that
$\|\mu_1-\mu_2\|_p=\|\alpha\|_{\ell^p}$ and $\linOp_{\mu_1}\varphi -\linOp_{\mu_2}\varphi=\linOp_\Lambda(\alpha,\varphi)$, for $\linOp_\Lambda (\,\cdot\, ,\varphi):\ell^p(\Lambda)\to M^p(\R)$ as defined in the statement of Proposition \ref{wellDefProp}.  With this,  \eqref{eq: identifiability set of operators} is equivalent to 
\begin{equation}\label{eq:main-thm-both-bd}
C_1 \big(\mrm{ms}(\Lambda_1,\Lambda_2) \wedge 1\big) \|\alpha\|_{\ell^p} \leq \|\m{H}_{\Lambda}(\alpha,\varphi)\|_{M^{p}(\R)}\leq  C_2 \|\alpha\|_{\ell^p},
\end{equation}
and hence it suffices to find constants $C_1=C_1(p,\msc{H},\varphi) >0$ and $C_2=C_2(p,\msc{H},\varphi)>0$ such that \eqref{eq:main-thm-both-bd} holds.
To this end, first note that by item (i) of Proposition \ref{wellDefProp} we have
\begin{equation*}
\|\m{H}_{\Lambda}(\alpha,\varphi) \|_{M^p(\R)}\lesssim_{\, \varphi}  \mrm{rel}(\Lambda)\|\alpha\|_p.
\end{equation*}
Furthermore, as $\mu_1,\mu_2\in \msc{M}_s^p$, we get
\begin{equation*}
\mrm{rel}(\Lambda_1\cup\Lambda_2)\leq \mrm{rel}(\Lambda_1) +\mrm{rel}(\Lambda_2)\lesssim s^{-2},
\end{equation*}
and so the upper bound in \eqref{eq:main-thm-both-bd} holds for some $C_2>0$ depending on $\varphi$ and $s$, as desired.

We proceed to establish the lower bound in \eqref{eq:main-thm-both-bd}. 
Note that this bound holds trivially if $\mrm{ms}(\Lambda_1,\Lambda_2)=0$ or $\Lambda_1=\Lambda_2= \varnothing$, so suppose w.l.o.g. that $\mrm{ms}(\Lambda_1,\Lambda_2)>0$ and $\Lambda_1\neq \varnothing$. Then, in  particular, $\Lambda\neq\varnothing$. Now, as $\linOp_\Lambda (\,\cdot\, ,\varphi):\ell^p(\Lambda)\to M^p(\R)$ is a continuous linear operator between Banach spaces, Lemma \ref{ShortTheoryLemma} implies that it suffices to find a $C_1=C_1(p,\msc{H},\varphi)>0$ such that the following statement holds:

(P1) \emph{
 For every $\beta\in \ell^q(\Lambda)$, there exists a $y\in M^q(\R)$ such that $\left(\linOp_\Lambda (\cdot,\varphi)\right)^*(y)=\beta$ and 
\begin{equation*}
C_1 \big(\mrm{ms}(\Lambda_1,\Lambda_2) \wedge 1\big) \|y\|_{M^q(\R)}\leq \|\beta\|_{\ell^q}.
\end{equation*}}
By item (ii) of Proposition \ref{wellDefProp} and Lemma \ref{STFT=Barg-lem}, we have the following expression for $\left(\linOp_\Lambda (\cdot,\varphi)\right)^*$ in terms of the Bargmann transform:
\begin{equation}\label{CritThm-Dual=Barg}
\left(\linOp_\Lambda (\cdot,\varphi)\right)^*(y)=\{e^{-\pi i \tau\nu}e^{-\pi|\lambda|^2/2}(\bargmann y)(\overline{\lambda})\}_{\lambda=\tau+i\nu \, \in \Lambda},\quad y\in M^q(\R).
\end{equation}
Thus, as the Bargmann transform is an isometric isomorphism between $M^p(\R)$ and $\m{F}^p(\C)$, and the map $\{\beta_\lambda\}_{\lambda=\tau+i\nu\, \in\Lambda}\mapsto \{\beta_\lambda e^{-\pi i \tau\nu}\}_{\lambda=\tau+i\nu \, \in \Lambda} $ is an isometric isomorphism on $\ell^q(\Lambda)$, the statement (P1) is equivalent to the following statement about interpolation:

(P2) \emph{
For every $\beta\in \ell^q(\Lambda)$, there exists an $F\in \m{F}^q(\C)$ such that $e^{-\pi |\lambda|^2/2 }F(\overline\lambda)=\beta_\lambda$, for all $\lambda\in\Lambda$, and 
\begin{equation} \label{eq:main-thm-P2}
C_1 \big(\mrm{ms}(\Lambda_1,\Lambda_2) \wedge 1\big) \|F\|_{\m{F}^q(\C)}\leq \|\beta\|_{\ell^q}.
\end{equation}
}
To prove (P2), we will make use of the interpolation basis functions provided by Lemma \ref{massiveLemma}. To this end, fix $\theta>0$ and $\gamma>0$ such that $2 \m{D}^+(\m{L})< 2\theta <\gamma^{-2}<1$, and let $\beta\in \ell^q(\Lambda)$ be arbitrary. Then, by Lemma \ref{lem:Rayleigh-number}, there exists an $R_0>0$ (depending only on $\msc{H}$) such that $n^+(\Lambda_{j},(0,R)^2)\leq \theta R^2$, for $j\in\{1,2\}$ and $R\geq R_0$. Now, for each $\lambda\in\Lambda$, define the set 
$
\wtd{\Lambda}_\lambda=\{\overline{\lambda'}-\overline{\lambda}: \lambda'\in \Lambda \}.
$
We will seek to apply Lemma \ref{massiveLemma} to each of the sets $\wtd{\Lambda}_\lambda$ as $\lambda$ ranges over $\Lambda$. To this end, first note that
\begin{equation*}
n^+(\wtd{\Lambda}_{\lambda},(0,R)^2)=n^+(\Lambda ,(0,R)^2) \leq n^+(\Lambda_{1},(0,R)^2)+n^+(\Lambda_{2},(0,R)^2)\leq 2 \theta R^2,
\end{equation*}
for all $R\geq R_0$. Therefore, as $\gamma<(2\theta)^{-1/2} $, it follows by Lemma \ref{lem:unif-lat-dist} that there exists an $R'=R'(\theta,\gamma, R_0)$ such that $\wtd{\Lambda}_\lambda$ is $R'$-uniformly close to $\Omega_\gamma=\{\omega_{m,n}= \gamma(m+in):m,n\in\Z\}$. In particular, there exists an enumeration $\wtd{\Lambda}_\lambda=\{\wtd{\lambda}_{m,n}\}_{(m,n)\in\wtd{\m{I}}}$ such that $|\wtd{\lambda}_{m,n}-\omega_{m,n}|\leq R'$, for all $(m,n)\in\wtd{\m{I}}$. Note that $0\in \wtd{\Lambda}_\lambda$ by definition of $\wtd{\Lambda}_\lambda$. In order to apply Lemma \ref{massiveLemma} we need to additionally ensure that we work with an enumeration of $\wtd{\Lambda}_\lambda=\{{\lambda}_{m,n}\}_{(m,n)\in{\m{I}}}$ (possibly different from the enumeration $\wtd{\Lambda}_\lambda=\{\wtd{\lambda}_{m,n}\}_{(m,n)\in\wtd{\m{I}}}$) that satisfies ${\lambda}_{0,0}=0$. To this end, let $(m_0,n_0)\in \wtd{\m{I}}$ be the index such that $\wtd{\lambda}_{m_0,n_0}=0$, and define $\m{I}$ and $\{{\lambda}_{m,n}\}_{(m,n)\in{\m{I}}}$ as follows:
\begin{itemize}
\item[--] If $(0,0)\notin\wtd{\m{I}}$, set $\m{I}=\big(\wtd{\m{I}}\setminus\{(m_0,n_0)\}\big) \cup\{(0,0)\} $, and let
\begin{equation*}
\lambda_{m,n}=\begin{cases}
0,& \text{if }(m,n)=(0,0),\\
\wtd{\lambda}_{m,n},\, & \text{if }(m,n)\in \m{I}\setminus \{(0,0)\}
\end{cases}.
\end{equation*}
\item[--] If $(0,0)\in\wtd{\m{I}}$, set $\m{I}=\wtd{\m{I}}$, and let 
\begin{equation*}
\lambda_{m,n}=\begin{cases}
0,& \text{if }(m,n)=(0,0),\\
\wtd{\lambda}_{0,0},\, & \text{if }(m,n)=(m_0,n_0),\\
\wtd{\lambda}_{m,n},\, & \text{if }\in (m,n)\in \m{I}\setminus \{(0,0),(m_0,n_0)\}
\end{cases}.
\end{equation*}
\end{itemize}
The new enumeration $\wtd{\Lambda}_\lambda=\{{\lambda}_{m,n}\}_{(m,n)\in{\m{I}}}$ satisfies $|\lambda_{0,0}-\omega_{0,0}|=0$ and $|\lambda_{m_0,n_0}-\omega_{m_0,n_0}|\leq |\wtd{\lambda}_{0,0}| + |\wtd{\lambda}_{m_0,n_0}-\omega_{m_0,n_0}|\leq 2R'$, and thus we have $|{\lambda}_{m,n}-\omega_{m,n}|\leq 2R'$, for all $(m,n)\in{\m{I}}$. The set $\wtd{\Lambda}_\lambda$ therefore satisfies the assumptions of Lemma \ref{massiveLemma} with $\rho\coleqq \mrm{ms}(\Lambda_1,\Lambda_2) \wedge\frac {s}{2}$, $s$, $\theta$, $\gamma$, and $R\coleqq R_0\vee (2R')$, and so the function $g_{-\lambda}\coleqq \wtd{g}_{\wtd{\Lambda}_\lambda}(\,\cdot - \overline{\lambda})$, where $\wtd{g}_{\wtd{\Lambda}_\lambda}$ is  defined according to \eqref{eq:massive-def-gtilde}, satisfies
\begin{equation}\label{eq:suff-thm-1}
g_{-\lambda}(\overline{\lambda'})=\begin{cases}
1,\, & \text{if }\lambda'=\lambda\\
0,\, & \text{if }\lambda'\neq \lambda
\end{cases}, \quad\text{for all } \lambda'\in\Lambda,
\end{equation}
and
\begin{equation}\label{eq:suff-thm-2}
|g_{-\lambda}(z)|
\leq C (\rho\wedge 1)^{-1} e^{\frac{\pi}{2}\gamma^{-2}|z-\overline{\lambda}|^2+c|z-\overline{\lambda}|\log{|z-\overline{\lambda}|}},\quad\text{for all }z\in \C,
\end{equation}
where $c>0$ and $C>0$ depend on $s$, $\theta$, $\gamma$, and $R$. Moreover, as $\lambda$ was arbitrary, \eqref {eq:suff-thm-1} and \eqref{eq:suff-thm-2} hold for all $\lambda\in\Lambda$.
Next, following \cite[p. 112]{Seip1992-2}, we consider the interpolation function
\begin{equation}\label{eq:main-thm-F-exp}
F(z)=\sum_{\lambda\in \Lambda}\beta_\lambda \;e^{\pi\lambda z-\frac{\pi}{2} |\lambda|^2}\;g_{-\lambda}(z).
\end{equation}
To see that $F$ is an element of $\m{F}^q(\C)$, observe that
\begin{align}
|F(z)|e^{-\frac{\pi}{2}|z|^2 }&\leq\sum_{\lambda\in\Lambda}|\beta_\lambda |e^{-\frac{\pi}{2}|z-\overline\lambda|^2}|g_{-\lambda}(z)|\notag\\
&\leq  C (\rho\wedge 1)^{-1} \sum_{\lambda\in\Lambda} |\beta_\lambda|e^{-\frac{\pi}{2}(1-\gamma^{-2})|z-\overline{\lambda}|^2+c|z-\overline{\lambda}|\log{|z-\overline{\lambda}|}}\notag\\
&=  C (\rho\wedge 1)^{-1}\sum_{j\in\{1,2\}}\sum_{\lambda\in\Lambda_j} |\beta_\lambda|f(z-\overline{\lambda}),\quad\text{for all }z\in \C,\notag
\end{align}
where
$
f(z)=\exp{\big[-\frac{\pi}{2}(1-\gamma^{-2})|z|^2+c|z|\log{|z|}\big]}
$.
Now, as $\gamma^{-2}<1$, we have that $f$ decays exponentially, and so $f\in W(L^\infty, L^1)$. Lemma \ref{IrregConvLem}  thus yields
\begin{equation*}
\Big\| \sum_{\lambda\in\Lambda_j} |\beta_\lambda|f(\cdot -\overline{\lambda})\Big\|_{L^q(\C)}\lesssim_{\, p,s}\|f\|_{W(L^\infty,L^1)} \|\{\beta_\lambda\}_{\lambda\in\Lambda_j}\|_{\ell^q},\quad\text{for } j\in\{1,2\},
\end{equation*}
and so 
\begin{equation}\label{eq:main-thm-last}
\|F\|_{\m{F}^q(\C)}=\big\| F(\cdot)e^{-\frac{\pi}{2}\abs{\cdot}^2} \big\|_{L^q(\C)}\leq C (\rho\wedge 1)^{-1}\Big\| \sum_{j\in\{1,2\}} \sum_{\lambda\in\Lambda_j} |\beta_\lambda|f(\cdot -\overline{\lambda})\Big\|_{L^q(\C)}\lesssim_{\,p,s,\gamma} C (\rho\wedge 1)^{-1} \|\beta\|_{\ell^q}.
\end{equation}
Now, recall that $\rho\coleqq \mrm{ms}(\Lambda_1,\Lambda_2) \wedge\frac {s}{2}$, and so $ \mrm{ms}(\Lambda_1,\Lambda_2) \wedge 1 \lesssim_{\, s}\rho\wedge 1 $. This together with \eqref{eq:main-thm-last} establishes \eqref{eq:main-thm-P2} with some $C_1>0$ depending on $s$, $\theta$, $\gamma$, $R_0$, and $R'$. As these quantities ultimately depend only on $s$ and $\msc{H}$, so does $C_1$. Finally, \eqref{eq:main-thm-F-exp} and the basis interpolation property \eqref{eq:suff-thm-1} together yield $F(\overline{\lambda})=e^{\pi |\lambda|^2/2}\beta_\lambda$, for all $\lambda\in\Lambda$. We have thus established (P2), thereby concluding the proof of the theorem.
\end{proof}


\section{Proof of Theorem \ref{NeceCritThm}}

In the proof of  Theorem \ref{NeceCritThm}  we will make use of the following results from \cite{Groechenig2015}, as well as a combinatorial lemma about squares in the plane, whose proof can be found in the appendix.

\begin{thm}[Non-uniform Balian-Low Theorem, {\cite[Cor. 1.2]{Groechenig2015}}] \label{GroBalianLow}
Let $\Lambda$ be a relatively separated subset of $\R^2$ and $x\in M^1(\R)$. If $\{\pi(\lambda) x\}_{\lambda\in\Lambda}$ is a Riesz sequence, i.e. $\|\sum_{\lambda\in\Lambda} c_\lambda \pi(\lambda) x\|_{L^2(\R)}\asymp_{\,\Lambda,x} \|c\|_{\ell^2(\Lambda)}$, for all $c\in\ell^2(\Lambda)$, then $D^+(\Lambda)<1$.
\end{thm} 
\begin{thm}[\!\!{\cite[Thm. 3.2]{Groechenig2015}}]\label{GroStab}
Let $\Lambda$ be a relatively separated subset of $\R^2$ and $x\in M^1(\R)$. Then $\|\sum_{\lambda\in\Lambda} c_\lambda \pi(\lambda) x\|_{M^p(\R)}\asymp_{\,\Lambda,x} \|c\|_{\ell^p(\Lambda)}$, $c\in\ell^p(\Lambda)$, holds for some $p\in[1,\infty]$ if and only if it holds for all $p\in[1,\infty]$.
\end{thm}
\begin{lem}[\!\!{\cite[Lem. 4.5]{Groechenig2015}}]\label{GroSetBA}
Let $\{\Lambda_n\}_{n\in\N}$ be a sequence of relatively separated subsets of $\R^2$. If $\sup_{n\in\N}\mrm{rel}(\Lambda_n)<\infty$, then there exists a subsequence $\{\Lambda_{n_k}\}_{k\in\N}$ that converges weakly to a relatively separated set.
\end{lem}
\begin{lem}\label{RecursiveSquareLemma}
Let $Y\subset \R^2$, $n\in\N$, and suppose that $K_n$ is a square in the plane of side length $\sqrt{2}(2^n+1)$ such that $\#(K_n\cap Y)\geq 2^{2n}+1$. Then there exist squares $K_0,K_1,\dots, K_{n-1}$ so that, for every $j\in\{0,1,\dots,n-1\}$, 
\begin{itemize}
\item[(i)] $K_j\subset K_{j+1}$, $K_j$ has sides of length $\sqrt{2}(2^j+1)$ parallel to the sides of $K_{j+1}$, and $K_j$ and $K_{j+1}$ share a corner,
\item[(ii)] $\# (K_j\cap Y)\geq 2^{2j}+1$.
\end{itemize}
We call a sequence $(K_0,K_1,\dots,K_n)$ satisfying (i) and (ii) a sequence of \textit{nested squares}.
\end{lem}

\begin{proof}[Proof of Theorem \ref{NeceCritThm}]
We argue by contradiction, so suppose that $\msc{H}(\m{L})^p$ is identifiable, but $\m{D}^+(\m{L})\geq \frac{1}{2}$. Define $\gamma_n=\sqrt{2}(1+2^{-n})$ and $R_n=2(2^n+1)$, for $n\in\N$. It then follows by Lemma \ref{lem:class-lat-dist} that, for every $n\in\N$, there exists a $\wtd{\mu}_n\in \msc{H}(\m{L})^p$ such that $\supp(\wtd{\mu}_n)$ is not $R_n$-uniformly close to $\Omega_{\gamma_n}$. Indeed, if this were not the case for some $n\in\N$, we would have $\m{D}^+(\m{L})\leq \gamma_n^{-2}<\frac{1}{2}$, contradicting our assumption that $\m{D}^+(\m{L})\geq \frac{1}{2}$. Fix such a $\wtd{\mu}_n$ for each $n$.

Now, for a fixed $n\in\N$, define the sets 
\begin{equation*}
S_{k,\ell}=\left[\sqrt{2}(2^n+1)k,\sqrt{2}(2^n+1)(k+1)\right)\times \left[\sqrt{2}(2^n+1)\ell,\sqrt{2}(2^n+1)(\ell+1)\right)\subset \R^2,
\end{equation*}
 for $(k,\ell)\in\Z^2$, forming a partition of the plane into squares of side length $\sqrt{2}(2^n+1)$. As every $S_{k,\ell}$ consists of exactly $2^{2n}$ fundamental cells of the lattice $\Omega_{\gamma_n}$ and the diagonal of $S_{k,\ell}$ has length $R_n$, there must exist a pair $(k_n,\ell_n)\in\Z^2$ such that $\#(S_{k_n,\ell_n}\cap \supp(\wtd{\mu}_n))\geq 2^{2n}+1$, for otherwise $\supp(\wtd{\mu}_n)$ would be $R_n$-uniformly close to $\Omega_{\gamma_n}$, contradicting our choice of $\wtd{\mu}_n$.

We set $\wtd{K}^n_n=S_{k_n,\ell_n}$ and apply Lemma \ref{RecursiveSquareLemma} with $\wtd{K}^n_n$ and $Y=\supp(\wtd{\mu}_n)$ to obtain a sequence $(\wtd{K}_0^n,\wtd{K}_1^n,\dots,\wtd{K}_n^n)$ of nested squares. Next, let $\lambda_n$ be the center of $\wtd{K}^n_0$ and note that, as $\m{L}$ is shift-invariant by assumption, there exists a measure $\mu_n\in \msc{H}(\m{L})^p$ with support $\supp(\wtd{\mu}_n )-\lambda_n$. Therefore, setting $K_j^n=\wtd{K}_j^n-\lambda_n$, we have that $(K_0^n,K_1^n,\dots,K_n^n)$ is a sequence of nested squares, $\# (K_j^n\cap \supp(\mu_n))\geq 2^{2j}+1$, and $K_0^n=\big[-\sqrt{2},\sqrt{2}\big)\times\big[-\sqrt{2},\sqrt{2}\big)$, for all $n\in\N$ and $j\in\{0,1,\dots n\}$. We next need to verify the following auxiliary claim.

\textit{Claim: Let $r \in\N$ and suppose that $\{\mu_{n_k^{r}}\}_{k\in\N }$ is a subsequence of $\{\mu_n\}_{n\in\N}$ such that $K_j^{n_j^{r}}=K_j^{n_k^{r}}$, for all $j\in\{1,2,\dots, r\}$ and all $k\geq j$. Then there exists a further subsequence $\{\mu_{n_k^{r+1}}\}_{k\in\N}$ such that $K_{r+1}^{n_{r+1}^{r+1}}=K_{r+1}^{n_k^{r+1}}$, for all $k\geq r+1$.}

\noindent \textit{Proof of Claim: } Let $\msc{K}$ be the set of squares $K'\subset K_{r}^{n_r^{r} }$ of side length $\sqrt{2}(2^{r+1}+1)$ such that $K'$ and $K_{r}^{n_r^{r} }$ have parallel sides and share a corner. As $(K_{0}^{n_k^r},K_1^{n_k^r},\dots,K_{r}^{n_k^r}, K_{r+1}^{n_k^r})$ is a sequence of nested squares, for all $k\geq r+1$, we have that $K_{r+1}^{n_k^r}\in\msc{K}$, for all $k\geq r+1$. But $\#\msc{K}=4$, and therefore at least one element of $\msc{K}$ appears infinitely often in the sequence $\{K_{r+1}^{n_k^r}\}_{k\geq r+1}$. We can therefore extract a subsequence $\{\mu_{n_k^{r+1}}\}_{k\in\N}$ of $\{\mu_{n_k^{r}}\}_{k\in\N }$ such that $K_{r+1}^{n_{r+1}^{r+1}}=K_{r+1}^{n_k^{r+1}}$, for all $k\geq r+1$, establishing the claim.

 Now, as $K_0^{n}=K_0^{0}$, for all $n\in\N$, we can apply a diagonalization argument together with the Claim to construct a subsequence $\{\mu_{n_k}\}_{k\in \N}$ of $\{\mu_n\}_{n\in\N}$ such that $K_j^{n_j}=K_j^{n_k}$ for all $j \in \N$ and all $k\geq j$.
Next, as $\msc{H}(\m{L})^p$ is $s$-regular, we have $ \inf_{\Lambda\in \m{L} }\mrm{sep}(\Lambda)\geq s >0$, and so $\sup_{k\in\N}\mrm{rel}(\supp(\mu_{n_k}))\lesssim s^{-1}<\infty$. Therefore, by passing to a further subsequence of $\{\mu_{n_k}\}_{k\in\N}$ if necessary, Lemma \ref{GroSetBA} implies the existence of a set $\Lambda^*\subset\R^2$ such that $\supp(\mu_{n_k})\xrightarrow[]{w} \Lambda^*$ as $k\to \infty$. Then,  as $\#(K_j^{n_j}\cap \supp(\mu_{n_j}) )\geq 2^{2j}+1$, we have 
\begin{equation*}
n^+\Big(\Lambda^*,[0, \sqrt{2}(2^j+1)+2]^2 \Big) \geq \#\Big(\Lambda^*\cap(K_j^{n_j}+B_1(0))\Big)\geq 2^{2j}+1,
\end{equation*}
 for all sufficiently large $j\in\N$, and so
\begin{equation}\label{eq:conv-thm-1/2}
D^+(\Lambda^*)\geq \limsup_{j\to\infty}\frac{2^{2j}+1}{\left(\sqrt{2}(2^j+1)+2 \right)^2}=\frac{1}{2}.
\end{equation}

Now, as $\m{L}$ is CSI, we have $\Lambda^*\in \m{L}$ and  $\Lambda^*+(\frac{s}{2},0)\in \m{L}$. Moreover, as $\sep(\supp(\mu_{n_k}))\geq s$, for all $k\in\N$, we have $\sep(\Lambda^*)\geq s$, and so $\mrm{ms}(\Lambda^*,\Lambda^*+(\frac{s}{2},0) )\geq\frac{s}{2}$. Therefore, as we assumed $\msc{H}(\m{L})^p$ to be identifiable, there exist $C_1,C_2>0$ depending on $p$ and $\m{L}$  and a probing signal $x\in M^1(\R)$ such that 
\begin{equation*}
C_1 \left(\frac{s}{2}\wedge 1\right)\norm{\nu}_{p} \leq \norm{ \sum_{\lambda\in\Lambda} \nu_\lambda \pi(\lambda) x }_{\MSp{p}{\R}} \leq C_2\norm{\nu}_{p},
\end{equation*}
for all $\nu\in\msc{M}^p$ supported on $\Lambda\coleqq \Lambda^*\cup\left(\Lambda^*+(\frac{s}{2},0)\right)$, and so by Theorems \ref{GroStab} and  \ref{GroBalianLow} we must have $D^+(\Lambda)<1$. On the other hand, \eqref{eq:conv-thm-1/2} implies
\begin{equation*}
D^+(\Lambda)=2D^+(\Lambda^*)\geq 1,
\end{equation*}
which stands in contradiction to $D^+(\Lambda)<1$. Our inital assumption must hence be false, concluding the proof of the theorem.
\end{proof}


\section{Proofs of Theorems  \ref{w*w*thm} and \ref{RobustConvThm}}

We start with the following proposition that quantifies the behavior of Riesz sums of time-frequency shifts of the probing signal $x$ under perturbation of the individual time-frequency shifts. We do so under a mild condition on the time-frequency spread of the probing signal. Concretely, $x$ will be assumed to be an element of the weighted modulation space $M^1_m(\R)=\{f\in \m{S}': \m{V}_\varphi f\in L^1_m(\R)\}$.

\begin{prop}\label{PerturLem}
Let $p\in[1,\infty)$, $\Lambda\subset\C$  a separated set, and let $\alpha \in\ell^p(\Lambda)$. Suppose that $x\in M^1_m(\R)$ with the weight function $m(z)=1+|z|$. Then
\begin{itemize}
\item[(i)] there exists a $\Phi\in W(L^\infty,L^1)$ depending only on $x$ such that
\begin{equation*}
|\m{V}_\varphi (x-\pi(\epsilon)x)|(u,v)\leq |\epsilon| \Phi(u,v),
\end{equation*}
for all $(u,v)\in\R^2$ and $\epsilon\in\C$ with $|\epsilon|\leq 1$.
\item[(ii)] \begin{equation*}
\Big\| \sum_{\lambda\in\Lambda}\alpha_\lambda \pi(\lambda) x-\sum_{\lambda\in\Lambda} \alpha_\lambda e^{2\pi i \Re (\lambda ) \Im(\varepsilon_\lambda) }\; \pi(\lambda+\varepsilon_\lambda)x \Big\|_{M^p(\R)}\lesssim_{\, p,s,x} \|\bve\|_{\ell^\infty (\Lambda)} \|\alpha\|_{\ell^p(\Lambda)},
\end{equation*}
for all $\bve =\{\varepsilon_\lambda\}_{\lambda\in\Lambda} \in\ell^\infty(\Lambda)$ such that $\|\bve\|_{\ell^\infty(\Lambda)}\leq 1$.
\end{itemize}
\end{prop}
\begin{proof}
(i) Fix an $\epsilon\in\C$ with $|\epsilon|\leq 1$. We split
\begin{equation}\label{PerturLemSplit}
|\m{V}_\varphi(x-\pi(\epsilon)x)|\leq |\m{V}_\varphi(x-\modu_{\Im(\epsilon) }x)|+|\m{V}_\varphi \modu_{\Im(\epsilon) }(x-\tra_{\Re(\epsilon)}x)|
\end{equation}
and bound each term on the right-hand side separately, beginning with the second term. To this end, we first define the auxiliary quantity $F (u,v)=(\m{V}_{\varphi}x)(u,v)e^{2\pi i uv}$. Then, for all $(u,v)\in\R^2$ and $\tau\in \R$, we have
\begin{equation*}
\begin{aligned}
|\m{V}_{\varphi}(x-\tra_\tau x)|(u,v)&=|(\m{V}_{\varphi} x)(u,v) - e^{-2\pi i \tau v} (\m{V}_{\varphi}x)(u-\tau,v) |=| F(u,v)- F(u-\tau,v)|\\
&=\abs{\int_{-\tau}^0 (\partial_{u}F)(u+r,v)\, \mrm{d}r} \leq |\tau| \sup_{r\in[-\tau,0]}|(\partial_u F)(u+r,v)|.
\end{aligned}
\end{equation*}
Therefore, for all $(u,v)\in\R^2$,
\begin{align}
|\m{V}_\varphi \modu_{\Im(\epsilon)}(x-\tra_{\Re(\epsilon) }x)|(u,v)&=|\m{V}_\varphi(x-\tra_{\Re(\epsilon) }x)|(u,v - \Im(\epsilon) ) \notag\\
&\leq |\Re(\epsilon)| \sup_{r\in[-\Re(\epsilon),0]}|(\partial_{u}F)(u+r,v- \Im(\epsilon))|\notag\\
&\leq |\epsilon| \;\cdot\hspace*{-1em} \sup_{\|(r_1,r_2)\|_{\infty}\leq 1}|(\partial_{u}F)(u+r_1,v+r_2)|,\label{PerturLem2ndTerm}
\end{align}
where we used the assumption $|\epsilon|\leq 1$.
Next, recalling the definition of $F$, we have
\begin{align}
|(\partial_u F)(u,v)|&=\left|\left( \partial_u (\m{V}_{\varphi}x)(u,v) \right)e^{2\pi i u v} +  (\m{V}_{\varphi}x)(u,v) \cdot \partial _u e^{2\pi i u v} \right| \notag\\
&=\left| \partial_u \langle x,\modu_v\tra_u\varphi\rangle +2\pi i v \langle x,\modu_v\tra_u\varphi\rangle\right|\notag\\
&\leq|(\m{V}_{\varphi'}x)(u,v)|+2\pi |v| |(\m{V}_\varphi x)(u,v)|,\label{PerturLemAbound}
\end{align}
where interchanging $\partial_u$ with the inner product in the last step is justified due to $\m{V}_{\varphi'}x\in L^1_m(\R^2)$ (which follows from \cite[Prop. 12.1.2]{Groechenig2000}). 
Therefore, \eqref{PerturLem2ndTerm} and \eqref{PerturLemAbound} together yield
\begin{equation}\label{PerturLem2ndTerm-bound}
|\m{V}_\varphi \modu_{\Im(\epsilon) }(x-\tra_{\Re(\epsilon) }x)|(u,v)\leq |\epsilon| \;\cdot\hspace*{-3mm}  \sup_{\|(r_1,r_2)\|_{\infty}\leq 1} \Psi(u+r_1,v+r_2),
\end{equation}
for all $(u,v)\in \R^2$, where we set
\begin{equation*}
\Psi(u,v)\coleqq |(\m{V}_{\varphi'}x)(u,v)|+2\pi \left(|u|+|v|\right) |(\m{V}_\varphi x)(u,v)|.
\end{equation*}
We bound the first term in \eqref{PerturLemSplit} in a similar manner, this time using another auxiliary quantity, namely $G (u,v)=(\m{V}_{\varphi}\widehat{x})(u,v)e^{2\pi i uv}$. Then, using the fundamental identity of time-frequency analysis \cite[eq. (3.10)]{Groechenig2000}
\begin{equation*}
(\m{V}_f  g)(u,v)=e^{-2\pi i u v}(\m{V}_{\widehat{f}}  \widehat{g})(v,-u), \qquad  f\in\m{S}, g\in\m{S}', (u,v)\in\R^2,
\end{equation*}
 and the fact that $\widehat{\varphi}=\varphi$, we obtain
\begin{align}
|\m{V}_\varphi(x-\modu_{\Im(\epsilon) }x)|(u,v)&=|\m{V}_{\widehat{\varphi}}(\widehat{x}-\tra_{\Im(\epsilon)}\widehat{x})|(v,-u)\notag\\
&\leq |\epsilon|\;\cdot\hspace*{-1em} \sup_{\|(r_1,r_2)\|_{\infty}\leq 1}|(\partial_{u}G)(v+r_1,-(u+r_2))|,\label{PerturLem1stTerm}
\end{align}
and
\begin{align}
|(\partial_u G)(v,-u)|&\leq|(\m{V}_{\varphi'}\widehat{x})(v,-u)|+2\pi |u| |(\m{V}_\varphi \widehat{x})(v,-u)|\label{PerturLemAhatbd-1}\\
&\leq|(\m{V}_{i \widehat{\varphi'}}\, \widehat{x})(v,-u)|+2\pi |u| |(\m{V}_{\widehat \varphi}\,  \widehat{x})(v,-u)|\label{PerturLemAhatbd-2}\\
&=|(\m{V}_{\varphi'}x)(u,v)|+2\pi |u| |(\m{V}_\varphi x)(u,v)|\label{PerturLemAhatbd}
\end{align}
for all $(u,v)\in\R^2$,  where \eqref{PerturLemAhatbd-1} is obtained analogously to \eqref{PerturLemAbound}, in \eqref{PerturLemAhatbd-2} we used $\varphi'= i \widehat{\varphi'}$, and in \eqref{PerturLemAhatbd} we again used the fundamental identity of time-frequency analysis.

Combining \eqref{PerturLem1stTerm} and \eqref{PerturLemAhatbd} thus yields
\begin{equation}\label{PerturLem1stTerm-bound}
|\m{V}_\varphi(x-\modu_{\Im(\epsilon)}x)|(u,v)\leq |\epsilon| \;\cdot\hspace*{-3mm}  \sup_{\|(r_1,r_2)\|_{\infty}\leq 1} \Psi(u+r_1,v+r_2),
\end{equation}
and so \eqref{PerturLem2ndTerm-bound} and \eqref{PerturLem1stTerm-bound} together give
\begin{equation*}
|\m{V}_\varphi(x-\pi(\epsilon)x)|(u,v)\leq 2|\epsilon|\;\cdot\hspace*{-3mm} \sup_{\|(r_1,r_2)\|_{\infty}\leq 1} \Psi(u+r_1,v+r_2).
\end{equation*}
Therefore, in order to complete the proof of item (i), it suffices to take 
\begin{equation*}
\Phi(u,v)= \sup_{\|(r_1,r_2)\|_{\infty}\leq 1} 2\Psi(u+r_1,v+r_2),
\end{equation*}
and show that $\Phi\in W(L^\infty,L^1)$. In fact, as
\begin{align}
\|\Phi\|_{W(L^\infty,L^1)}&\lesssim \int_{\R^2}\sup_{\|(y_1,y_2)\|_{\infty}\leq 1} \hspace*{-1em}\Phi(u+y_1,v+y_2)\,\mrm{d}u \mrm{d}v \leq 2\int_{\R^2}\sup_{\|(r_1,r_2)\|_{\infty}\leq 2} \hspace*{-1em}\Psi(u+r_1,v+r_2)\,\mrm{d}u\mrm{d}v\notag\\
&\lesssim \|\Psi\|_{W(L^\infty,L^1)}\notag,
\end{align}
it suffices to establish that $\Psi\in W(L^\infty,L^1)$. To this end, note that $\|\m{V}_{\varphi'}x\|_{L^1(\R^2)}\asymp \|x\|_{M^1(\R)}$ (see \cite[Prop. 11.4.2]{Groechenig2000}), and so by \cite[Prop. 12.1.11]{Groechenig2000} we have $\m{V}_{\varphi'}x\in W(L^\infty,L^1)$. Next, as $\varphi\in M^1_m(\R)$ and $x\in M^1_m(\R)$, we have by \cite[Prop. 12.1.11]{Groechenig2000} that $\m{V}_\varphi x\in W(L^\infty,L^1_m)$. Therefore, using $|u|+|v|\lesssim 1 + |u+iv|=m(u+iv)$, we have
\begin{equation*}
\|\Psi\|_{W(L^\infty,L^1)}\lesssim \|\m{V}_{\varphi'}x \|_{W(L^\infty,L^1)}+\|\m{V}_{\varphi}x \|_{W(L^\infty,L^1_m)}<\infty,
\end{equation*}
as desired.

\noindent(ii)  Recalling the definition of $\|\cdot\|_{M^p(\R)}$, we have
\begin{align}
&\Big\| {\sum_{\lambda\in\Lambda}\alpha_\lambda \pi(\lambda) x-\sum_{\lambda\in\Lambda} \alpha_\lambda e^{2\pi i \Re(\lambda)\Im(\varepsilon_\lambda)}\; \pi(\lambda+\varepsilon_\lambda)x} \Big\|_{M^p(\R)}\notag\\
=\;&\Big\| {\sum_{\lambda\in\Lambda}\alpha_\lambda  \m{V}_\varphi( \pi(\lambda) x)-\sum_{\lambda\in\Lambda} \alpha_\lambda e^{2\pi i \Re(\lambda)\Im(\varepsilon_\lambda)}\;  \m{V}_\varphi(\pi(\lambda+\varepsilon_\lambda)x) } \Big\|_{L^p(\R^2)}\notag\\
=\;&\Big\| {\sum_{\lambda\in\Lambda}\alpha_\lambda\, \m{V}_\varphi( \pi(\lambda)(x-\pi(\varepsilon_\lambda)x))} \Big\|_{L^p(\R^2)},\notag
\end{align}
where we used the commutation relation $\pi(\lambda+\varepsilon_\lambda)=e^{-2\pi i \Re(\lambda)\Im(\varepsilon_\lambda)}\pi(\lambda)\pi(\varepsilon_\lambda)$. Now, by item (i) we have 
\begin{equation*}
\abs{\m{V}_\varphi( \pi(\lambda)(x-\pi(\varepsilon_\lambda)x)) }=\abs{\m{V}_\varphi(x-\pi(\varepsilon_\lambda)x)}(\,\cdot-\lambda)\leq \|\bve\|_{\ell^\infty(\Lambda)} \Phi(\,\cdot-\lambda)
\end{equation*}
pointwise, for all $\lambda\in\Lambda$, and so, by Lemma \ref{IrregConvLem}, we find that
\begin{equation*}\label{PerturLemFin}
\Big\|{\sum_{\lambda\in\Lambda}\alpha_\lambda\, \m{V}_\varphi ( \pi(\lambda)(x-\pi(\varepsilon_\lambda)x) )}\Big\|_{L^p(\R^2)}\hspace*{-3mm}\leq\Big\|{\sum_{\lambda\in\Lambda}|\alpha_\lambda| \|\bve\|_{\ell^\infty(\Lambda)} \Phi (\,\cdot -\lambda)}\Big\|_{L^p(\R^2)} \lesssim_{\, p,s,x} \|\bve\|_{\ell^\infty(\Lambda)} \|\alpha\|_{\ell^p(\Lambda)}.
\end{equation*}
This establishes (ii) and completes the proof.
\end{proof}

We next show that weak-* convergence of measures $\mu_n\in \msc{M}^p_s$ implies weak-* convergence of the measurements $\linOp_{\mu_n}x\in M^p(\R)$.

\begin{prop}\label{ForConProp}
Let $p\in(1,\infty)$, $x\in M^1_m(\R)$ with the weight function $m(z)=1+|z|$, and let $\{\mu_n\}_{n\in\N}\subset \msc{M}^p_s$ be a sequence converging to some $\mu\in \msc{M}^p_s$ in the weak-* topology $W(C_0,L^q)$. Then $\linOp_{\mu_n}x\to \linOp_{\mu}x$ in the weak-* topology of $M^p(\R)$.
\end{prop}
\begin{proof}
As $p\in(1,\infty)$, $M^p(\R)$ is reflexive and so its weak and weak-* topologies coincide \cite[Thm. 11.3.6]{Groechenig2000}. Due to the dual pairing \eqref{MSDualPairing}, this topology is generated by the linear functionals $\langle \,\cdot\, , y\rangle_{M^p(\R)\times M^q(\R)}$, for $y\in M^q(\R)$, and so we have to show that
\begin{equation}\label{eq:ForConProp-1}
\lim_{n\to\infty} \langle \linOp_{\mu_n}x, y\rangle_{M^p(\R)\times M^q(\R)}= \langle \linOp_{\mu}x, y\rangle_{M^p(\R)\times M^q(\R)},\quad\text{for all }y\in M^q(\R).
\end{equation}
Now, for $y\in M^q(\R)$, set $f_y(\lambda)\coleqq \langle y,\pi(\lambda)x\rangle_{M^q(\R)\times M^p(\R)}$, for $\lambda=\tau+i\nu$. If we show that $f_y\in C_0$, we will then have
\begin{equation}\label{ForConPropEq1}
\langle \linOp_{\mu_n}x, y\rangle_{M^p(\R)\times M^q(\R)}= \sum_{\lambda\in\supp(\mu_n)}\mu_n(\{\lambda \})\langle \pi(\lambda)x,y\rangle_{M^p(\R)\times M^q(\R)}= \langle \mu_n, f_y\rangle_{W(\m{M},L^p)\times W(C_0,L^q)},
\end{equation}
since the dual pairing is continuous in its first argument, and so, as $\mu_{n}\xrightarrow[]{w^*} \mu$ by assumption, \eqref{ForConPropEq1} will imply \eqref{eq:ForConProp-1}. Therefore, in order to complete the proof it suffices to show that $f_y\in C_0$, for all $y\in M^q(\R)$.

To this end, fix an arbitrary $y\in M^q(\R)$ and note that then $\m{V}_x \, y\in W(L^\infty, L^q)$ by \cite[Thm. 12.2.1]{Groechenig2000}. On the other hand, H\"{o}lder's inequality yields
\begin{align}
\abs{f_y(\lambda)}=\abs{\langle y,\pi(\lambda)x\rangle_{M^q(\R)\times M^p(\R)}}&\leq \int_{\R^2}\abs{(\m{V}_\varphi y)(u,v)}\abs{(\m{V}_\varphi \pi(\lambda)x)(u,v)}\,\mrm{d}u\mrm{d}v\notag\\
&\leq\int_{\R^2}\abs{(\m{V}_\varphi y)(u,v)}\abs{(\m{V}_\varphi x)(u-\tau,v-\nu)}\,\mrm{d}u\mrm{d}v\notag\\
&=(|\m{V}_\varphi y|\ast |\m{V}_\varphi x (-\,\cdot\, )|)(\tau,\nu).\label{eq:ForConProp-2}
\end{align}
Next, as $x\in M^1(\R)$, we have $\m{V}_\varphi x\in W(L^\infty,L^1)$ by \cite[Thm. 12.2.1]{Groechenig2000}, which, together with $\m{V}_\varphi y\in L^q(\R^2)$ and \eqref{eq:ForConProp-2}, implies  by \cite[Thm. 11.1.5]{Groechenig2000} that $f_y\in W(L^\infty,L^q)$. This, in particular, shows that $f_y(z)\to 0$ as $|z|\to\infty$.
Consider now arbitrary $\lambda\in\C$  and $\epsilon\in\C$ with $|\epsilon|\leq 1$. Then, by applying item (i) of Proposition \ref{PerturLem} with $x$ replaced by $\pi(\lambda)x$, we find a $\Phi_\lambda \in W(L^\infty,L^1)\subset W(L^\infty, L^p)$ depending only on $x$ and $\lambda$ such that
\begin{equation*}
|\m{V}_{\varphi}(\pi(\lambda)x-\pi(\epsilon)\pi(\lambda)x)|\leq |\epsilon| \cdot \Phi_\lambda
\end{equation*}
pointwise. We thus have
\begin{align}
\|\pi(\lambda)x-\pi(\lambda+\epsilon)x\|_{M^p(\R)}&=\|\pi(\lambda)x-e^{2\pi i \Re(\epsilon)\Im(\lambda)}\pi(\epsilon)\pi(\lambda)x\|_{M^p(\R)}\notag\\
&=\|\pi(\lambda)x\|_{M^p(\R)}\abs{1-e^{2\pi i \Re(\epsilon)\Im(\lambda)}}+\|\pi(\lambda)x-\pi(\epsilon)\pi(\lambda)x\|_{M^p(\R)}\notag\\
&\leq \|\pi(\lambda)x\|_{M^p(\R)}\abs{1-e^{2\pi i \Re(\epsilon)\Im(\lambda)}}+ |\epsilon|\cdot  \|\Phi_\lambda\|_{L^q(\R^2)}\notag,
\end{align}
and so $\lim_{\epsilon\to 0}\|\pi(\lambda+\epsilon)x\to\pi(\lambda)x\|_{M^p(\R)}$. Therefore, by the continuity of the dual pairing in the second argument, we get
\begin{equation*}
f_y(\lambda+\epsilon)=\langle y,\pi(\lambda+\epsilon)x\rangle_{M^q(\R)\times M^p(\R)}\to \langle y,\pi(\lambda)x\rangle_{M^q(\R)\times M^p(\R)}=f_y(\lambda)\quad\text{as }\epsilon\to 0,
\end{equation*}
and so, as $\lambda$ was arbitrary, we deduce that $f_y$ is continuous. We have hence established that $f_y\in C_0$, completing the proof.
\end{proof}

We are now ready to prove Theorems \ref{w*w*thm} and \ref{RobustConvThm}. In addition to Propositions \ref{PerturLem} and \ref{ForConProp}, we will need the Banach-Alaoglu theorem as well as the inequality \eqref{eq:WnormEquiv-0}.

\begin{proof}[Proof of Theorem \ref{w*w*thm}] (i) Let $\mu,\wtd{\mu}\in\msc{H}(\m{L})^p$ be such that $\linOp_{\wtd\mu}\,x= \linOp_{\mu}x$, write ${\mu}=\sum_{\lambda\in\Lambda}\alpha_{\lambda}\delta_{{\lambda}}$, $\wtd{\mu}=\sum_{\wtd\lambda\in\wtd\Lambda}\alpha_{\wtd\lambda}\delta_{\wtd{\lambda}}$, and for $R>0$ define
\begin{align}
\delta_R&=\min\left\{|\lambda-\wtd{\lambda}|:\lambda\in \Lambda\cap B_R(0),\; \wtd{\lambda}\in\wtd{\Lambda}\setminus \Lambda\right\}\wedge \frac{s}{2}\wedge 1,\quad\text{and}\notag\\
\wtd{\Lambda}_R&=\{\wtd{\lambda}\in\wtd{\Lambda}\setminus \Lambda:|\wtd{\lambda}|>R+\delta_R,\, d(\wtd{\lambda},\Lambda)\leq \delta_R\}\notag.
\end{align}
Informally, $\delta_R$ is the distance between $\Lambda$ and $\wtd{\Lambda}$ restricted to the disk $B_R(0)$ (not counting the points in $\Lambda\cap\wtd \Lambda$), and $\wtd{\Lambda}_R$ is the part of the support of $\wtd{\mu}$ which is at least $R+\delta_R$ away from the origin and everywhere within $\delta_R$ of $\Lambda$. 
Fix an $R>0$, and, for $\wtd{\lambda}\in\wtd{\Lambda}_R$, write $\lambda(\wtd{\lambda})$ for the point of $\Lambda$ such that $|\wtd{\lambda}-\lambda(\wtd{\lambda})|\leq\delta_R$. Note that this point is unique, as $\delta_R\leq s/2$. Next, define the measures
\begin{align}
\wtd\mu_R^1&=\sum_{\wtd{\lambda}\in\wtd{\Lambda}_R}\alpha_{\,\wtd{\lambda}}\; e^{2\pi i\,  \Im(\lambda(\wtd\lambda)-\wtd{\lambda})\Re(\wtd\lambda) }\delta_{\lambda(\wtd{\lambda})},\quad\text{and}\notag\\
\wtd\mu_R^2&=\sum_{\wtd{\lambda}\in\wtd\Lambda\setminus\wtd{\Lambda}_R}\alpha_{\, \wtd{\lambda}}\,\delta_{\,\wtd{\lambda}}\notag
\end{align}
and note that $\wtd\mu_R^1$ is the measure obtained by ``shifting the support'' of the restricted measure $\wtd\mu\mathds{1}_{\wtd\Lambda_R}$ onto $\Lambda$, and $\wtd\mu_R^2$ is the remaining part of $\wtd{\mu}$.

Next, we have $\mrm{ms}(\supp(\mu-\wtd{\mu}^1_R),\supp(\wtd{\mu}^2_R))\geq \delta_R$.
Now, as $\supp(\mu-\wtd\mu_R^1)\subset\supp(\mu)$ and $\supp(\wtd{\mu}_R^2)\subset \supp(\wtd{\mu})$, we have that  $\mu-\wtd\mu_R^1$ and  $\wtd{\mu}_R^2$ are elements of $\msc{H}(\m{L})^p$, and so the identifiability condition \eqref{eq: identifiability set of operators} can be applied to the measures $\mu-\wtd\mu_R^1$ and $\wtd{\mu}_R^2$, yielding
\begin{align}
C_1(\delta_R\wedge 1)\|\mu-\wtd\mu_R^1-\wtd{\mu}_R^2\|_p&\leq \|\linOp_\mu x-\linOp_{\wtd\mu_R^1}x-\linOp_{\wtd{\mu}_R^2}x\|_{M^p(\R)}\notag\\
&=\|\linOp_{\wtd\mu} x-\linOp_{\wtd\mu_R^1}x-\linOp_{\wtd{\mu}_R^2}x\|_{M^p(\R)}\notag\\
&=\Big\| {\sum_{\wtd{\lambda}\in\wtd{\Lambda}_R}\alpha_{\,\wtd\lambda}\,\pi(\wtd\lambda)x
-\sum_{\wtd{\lambda}\in\wtd{\Lambda}_R}\alpha_{\,\wtd{\lambda}}\,e^{2\pi i\, \Im(\lambda(\wtd\lambda)-\wtd{\lambda})\Re(\wtd\lambda) }\pi({\lambda(\wtd{\lambda})})x
}\Big\|_{M^p(\R)}\notag\\
&\lesssim_{\, p,s,x} \,\delta_R\cdot \|\{\alpha_{\lambda}\}_{\lambda\in\wtd\Lambda_R}\|_{\ell^p}\notag ,
\end{align}
where in the last step we used item (ii) of Proposition \ref{PerturLem} with $\bve=\{\lambda(\wtd{\lambda}) - \wtd{\lambda}\}_{\wtd{\lambda}\in\wtd{\Lambda}_R}$, noting that $|\lambda(\wtd{\lambda}) - \wtd{\lambda}|\leq \delta_R\leq 1$, for all $\wtd{\lambda}\in\wtd{\Lambda}_R$. Therefore, by dividing both sides by $\delta_R$, we obtain
\begin{equation*}
\|\mu-\wtd\mu_R^1-\wtd{\mu}_R^2\|_p\lesssim_{\,p,\msc{H}(\m{L})^p,x  }\|\{\alpha_{\lambda}\}_{\lambda\in\wtd\Lambda_R}\|_{\ell^p}.
\end{equation*}
Now, as $R>0$ was arbitrary and $ \|\{\alpha_{\lambda}\}_{\lambda\in\wtd\Lambda_R}\|_{\ell^p}\to 0$ as $R\to \infty$, we deduce that $\|\mu-\wtd\mu_R^1-\wtd{\mu}_R^2\|_p\to 0$ as $R\to \infty$. Moreover, as $\|(\mu-\wtd\mu)\mathds{1}_{B_{R}(0)}\|_p\leq \|\mu-\wtd\mu_R^1-\wtd{\mu}_R^2\|_p$, we obtain $\|(\mu-\wtd\mu)\mathds{1}_{B_{R}(0)}\|_p\to 0$ as $R\to \infty$, which implies $\wtd{\mu}=\mu$ and hence completes the proof of (i).

\noindent(ii) The ``if'' direction follows immediately by Proposition \ref{ForConProp}. To show the ``only if'' direction, suppose that $\linOp_{\mu_n}x\to \linOp_{\mu}x$ in the weak-* topology of $M^p(\R)$.
It then suffices to establish that every subsequence of $\{\mu_n\}_{n\in\N}$ has a further subsequence that converges to $\mu$ in the weak-* topology of $\msc{M}_s^p$. To this end, fix an arbitrary subsequence $\{\mu_{n_k}\}_{k\in\N}$  of $\{\mu_n\}_{n\in\N}$ and let $\Lambda=\supp(\mu)$ and $\Lambda_k=\supp(\mu_{n_k})$. We then have $\limsup_{n}\mrm{rel}(\Lambda_k)\lesssim s^{-2}<\infty$, and so by \cite[Lem. 4.5]{Groechenig2015}, $\{\Lambda_k\}_{k\in\N}$ has a subsequence $\{\Lambda_{k_\ell}\}_{\ell\in\N}$ that converges weakly to a relatively separated set $\wtd{\Lambda}$. Note that, as $\sep(\Lambda_k)\geq s$, for all $k\in\N$, we also have $\sep(\wtd{\Lambda})\geq s$. Now, using \eqref{eq:WnormEquiv-0} and the identifiability condition \eqref{eq: identifiability set of operators}, we have
\begin{equation*}
\sup_{k\in\N}\|\mu_{n_{k_\ell}}\|_{W(\m{M},L^p)}\lesssim_{\,p,s} \sup_{k\in\N} \|\mu_{n_{k_\ell}}\|_{p}\lesssim_{\,p, \msc{H}(\m{L})^p,x } \; \sup_{k\in\N} \|\linOp_{\mu_{n_{k_\ell}}}x\|_{M^p(\R)}<\infty,
\end{equation*}
and therefore, by the Banach-Alaoglu theorem \cite[Thm. 3.15]{Rudin1991}, $\{\mu_{n_{k_\ell}}\}_{\ell\in\N}$ has a subsequence, which we will w.l.o.g. also denote by $\{\mu_{n_{k_\ell}}\}_{\ell\in\N}$ to lighten notation, such that $\mu_{n_{k_\ell}}\xrightarrow[]{w^*}  \wtd{\mu}$, for some $\wtd{\mu}\in W(\m{M},L^p)$. Note that then $\supp(\wtd{\mu})\subset \wtd{\Lambda}$ as $\Lambda_{k_\ell}\xrightarrow[]{w} \wtd\Lambda $, and so $\wtd{\mu}\in \msc{H}(\m{L})^p $.
It remains to show that $\wtd{\mu}=\mu$. To this end, note that by Proposition \ref{ForConProp} we have $\linOp_{\mu_n}x\to \linOp_{\wtd{\mu}}x$ in the weak-* topology of $M^p(\R)$, and so by uniqueness of weak-* limits we deduce that $\linOp_{\wtd{\mu}}x=\linOp_{{\mu}}x$.  From this it follows by item (i) of the theorem that $\wtd{\mu}=\mu$, which finishes the proof.
\end{proof}

\begin{proof}[Proof of Theorem \ref{RobustConvThm}]
Fix $\epsilon>0$ and an arbitrary finite subset $\wtd\Lambda\subset\Lambda$ such that $\|\mu-\wtd\mu\|_{p}<\epsilon$, where we set $\wtd\mu=\mu\mathds{1}_{\wtd\Lambda}$, and write $\wtd\mu=\sum_{\lambda\in\wtd\Lambda}\alpha_\lambda\delta_\lambda$. Then $\|\linOp_{\wtd\mu}x-\linOp_{\mu}x\|_{M^p(\R)}<C_2\epsilon$ by the identifiability condition. Next, as $\linOp_{\mu_n}x\to\linOp_{\mu}x$ in norm, this convergence also holds in the weak-* topology, and so by Theorem \ref{w*w*thm} we have $\mu_n\xrightarrow[]{w^*} \mu$. We can therefore decompose each $\mu_n$ according to $\mu_n=\nu_n+\rho_n$, where
\begin{equation*}
\nu_n=\sum_{\lambda\in{\wtd\Lambda}}\alpha_\lambda^{(n)}\delta_{\lambda+\bve_n(\lambda)},
\end{equation*}
and $\lim_{n\to\infty}\alpha_\lambda^{(n)}=\alpha_\lambda$, $\lim_{n\to\infty}|\bve_n(\lambda)|= 0$, for every $\lambda\in\wtd{\Lambda}$.
Now, for every $n\in\N$, define the ``shifted'' measure
\begin{equation*}
\wtd{\nu}_n=\sum_{\lambda\in{\tilde\Lambda}}\alpha_\lambda^{(n)} e^{-2\pi i\, \Re(\lambda) \Im(\varepsilon_\lambda)}\,\delta_\lambda.
\end{equation*}
Then, as $\wtd{\Lambda}$ is finite, we have $\lim_{n\to\infty}\|\bve_n\|_{\ell^\infty ({\wtd\Lambda})}=0$ and $\sup_{n\in\N} \|\alpha^{(n)}\|_{\ell^p({\wtd\Lambda})} <\infty$, which together with item (ii) of Proposition \ref{PerturLem} yields
\begin{align}
\|\linOp_{\wtd{\nu}_n}x-\linOp_{\nu_n}x\|_{M^p(\R)}&=\Big\| {\sum_{\lambda\in{\wtd\Lambda}}\alpha_\lambda^{(n)} e^{-2\pi i\, \Re(\lambda) \Im(\varepsilon_\lambda)}\pi(\lambda)x-
\sum_{\lambda\in{\wtd\Lambda}}\alpha_\lambda^{(n)}\pi(\lambda+\bve_n(\lambda))x} \Big\|_{M^p(\R)}\notag\\
&\lesssim_{\,p,s,x} \|\bve_n\|_{\ell^\infty ({\wtd\Lambda})} \|\alpha^{(n)}\|_{\ell^p({\wtd\Lambda})}\to 0 \quad\text{as }n\to\infty.\label{eq:robust-conv-thm-1}
\end{align}
 
To bound $\|\rho_n\|_p$, we employ the identifiability condition \eqref{eq: identifiability set of operators} with the measures $\rho_n$ and $\wtd\mu-\wtd\nu_n$. Concretely, we note that $\supp(\rho_n)\subset\supp(\mu_n)$ and $\supp(\wtd\mu-\wtd\nu_n)\subset\supp(\wtd{\mu})\subset\supp(\mu)$, and so $\rho_n$ and $\wtd\mu-\wtd\nu_n$ are indeed elements of $\msc{H}(\m{L})^p$. Now, $\mrm{ms}(\supp(\wtd\mu-\wtd\nu_n),\supp(\rho_n))\geq s-\|\bve_n(\lambda)\|_{\ell^\infty(\wtd{\Lambda})}>s/2$, for sufficiently large $n$, and so, by combining \eqref{eq:robust-conv-thm-1}, $\|\linOp_{\wtd\mu}x-\linOp_{\mu}x\|_{M^p(\R)}<C_2\epsilon$, and the assumption $\lim_{n\to\infty} \|\linOp_{\mu_n}x-\linOp_{\mu}x\|_{M^p(\R)}= 0$, we get
\begin{align}
C_1 \left(\frac{s}{2}\wedge 1\right)\|\wtd\mu-\wtd\nu_n-\rho_n\|_p&\leq\|\linOp_{\wtd\mu} x-\linOp_{\wtd\nu_n}x-\linOp_{\rho_n}x\|_{M^p(\R)}\notag\\
&=\| \linOp_{\wtd\mu} x- \linOp_{\wtd\nu_n}x - \linOp_{\mu_n}x + \linOp_{\nu_n}x\|_{M^p(\R)}\notag\\
&\leq \|\linOp_{\wtd\mu}x -\linOp_{\mu}x \|_{M^p(\R)}+ \|\linOp_{\mu}x -\linOp_{\mu_n}x \|_{M^p(\R)}+\|\linOp_{\nu_n}x -\linOp_{\wtd\nu_n}x \|_{M^p(\R)}\notag\\
&<2C_2\epsilon
\end{align}
for large enough $n$. On the other hand, $\|\wtd\mu-\wtd\nu_n-\rho_n\|_p\geq \|\rho_n\|_p$, and so
\begin{equation*}
\|\rho_n\|_p\leq \frac{2C_2}{C_1(s/2\wedge 1)}\epsilon\leq \frac{4C_2}{C_1(s\wedge 1)}\epsilon,
\end{equation*}
for all sufficiently large $n$. This concludes the proof.
\end{proof}


\section{Proofs of Proposition \ref{prop:examples-CSI} and Corollaries  \ref{cor:fin-class}, \ref{cor:SuffSepClass}, and \ref{cor:RayClass}}
 
\begin{proof}[Proof of Proposition \ref{prop:examples-CSI}]
Let $\m{L}$ be any of the classes $\m{L}_{s}^{\text{sep}}$, $\m{L}_{s,N}^{\text{fin}}$, and $\m{L}_{s,\theta,R}^{\text{Ray}}$. It then follows directly from the definitions of these sets that $\Lambda+z\coleqq \{\lambda+z :\lambda\in\Lambda\}\in\m{L}$, for all $\Lambda\in\m{L}$ and $z\in\R^2$, so it remains to verify closure under weak convergence. To this end, let $\{\Lambda_n\}_{n\in\N}$ be a sequence in $\m{L}$ such that $\Lambda_n\xrightarrow[]{w} \Lambda$ as $n\to\infty$, for some $\Lambda\subset\R^2$. In all three cases we have that $\inf_{n\in\N}\sep(\Lambda_n)\geq s$ implies $\sep(\Lambda)\geq s$, which establishes that $\m{L}_{s}^{\text{sep}}$ is CSI. 

Suppose now that $\m{L}=\m{L}_{s,N}^{\text{fin}}$. It then suffices to show that $\#\Lambda\leq N$. To this end, let $\wtd{\Lambda}$ be an arbitrary finite subset of $\Lambda$, and consider a point $\lambda\in\wtd\Lambda$. Then, for all sufficiently large $n$, there exists a $\lambda_n\in\Lambda_n$ such that $|\lambda-\lambda_n|<(\sep(\wtd{\Lambda})\wedge s)/2$. We deduce that, for all sufficiently large $n$, there exists a $\wtd{\Lambda}_n\subset\Lambda_n$ such that $\#\wtd\Lambda_n=\#\wtd\Lambda$. But $\#\Lambda_n\leq N$, for all $n\in\N$, by definition of the class $\m{L}_{s,N}^{\text{fin}}$, so we must have $\#\wtd\Lambda\leq N$. Now, as $\wtd{\Lambda}$ was arbitrary, we must have $\#\Lambda\leq N$, and so $\Lambda\in \m{L}_{s,N}^{\text{fin}}$. This establishes that $\m{L}_{s,N}^{\text{fin}}$ is CSI.

For $\m{L}=\m{L}_{s,\theta,R}^{\text{Ray}}$, fix an arbitrary translate $K^{\circ}_{x,y}=(x,x+R)\times(y,y+R)$ of $(0,R)^2$, and consider a point $\lambda\in\Lambda\cap K^{\circ}_{x,y}$. Then, as $K^{\circ}_{x,y}$ is open, we have that for all sufficiently large $n$ there exists a $\lambda_n\in\Lambda_n$ such that $|\lambda-\lambda_n|<s/2$ and $\lambda_n\in K^{\circ}_{x,y}$. Therefore, as $\lambda$ was arbitrary, and $\Lambda\cap K^{\circ}_{x,y}$ is finite, we have $\#(\Lambda\cap K^{\circ}_{x,y}) \leq \#(\Lambda_n\cap K^{\circ}_{x,y})\leq \theta R^2$, for sufficiently large $n$. Thus, as $K^{\circ}_{x,y}$  was arbitrary, we obtain $n^+ \left(\Lambda,(0,R)^2\right)\leq \theta R^2$, and so $\Lambda\in \m{L}_{s,\gamma,R}^{\text{Ray}}$. This establishes that $\m{L}_{s,\gamma,R}^{\text{Ray}}$ is CSI and thereby completes the proof.
\end{proof}

\begin{proof}[Proof of Corollary \ref{cor:fin-class}]
The proof is effected by verifying the conditions of Theorem \ref{CritThm}. As $s \coleqq { \inf_{\Lambda\in\m{L}_{s,N}^{\text{fin}}}} \, \sep(\Lambda)>0$, by definition of $\m{L}_{s,N}^{\text{fin}}$, it suffices to show that $\m{D}^+(\m{L}_{s,N}^{\text{fin}})<\frac{1}{2}\,$. To this end, fix a $\Lambda=\{\lambda_1,\dots,\lambda_n\} \in \m{L}_{s,N}^{\text{fin}}$, where $n\leq N$. Set $R=2 \sqrt{2}\lceil \frac{N+1}{2}\rceil$, and for $k\in\{1,\dots, n\}$ let $S_k=\{\xi\in\Omega_2: |\xi-\lambda_k|\leq R\}$, where we recall that $\Omega_2=\{2(m+in):m,n\in\Z\}$. Now, for every $k\in\{1,\dots,n\}$, choose a $\xi_k\in S_k\setminus\{\xi_1,\dots,\xi_{k-1}\}$. Note that this is possible as $\# S_k\geq N$, for all $k$. We thus have $|\lambda_k-\xi_k|\leq R$ for all $k$, and so, as $\Lambda$ was arbitrary, we have that $\Lambda$ is $R$-uniformly close to $\Omega_2$, for all $\Lambda\in \m{L}_{s,N}^{\text{fin}}$. Lemma \ref{lem:class-lat-dist} therefore implies $\m{D}^+(\m{L}_{s,N}^{\text{fin}})\leq 2^{-2}<\frac{1}{2}$, completing the proof.
\end{proof}

\begin{proof}[Proof of Corollary \ref{cor:SuffSepClass}]
First assume that $s>2\cdot 3^{-\frac{1}{4}}$. In view of Theorem \ref{CritThm}, it again suffices to verify that $\m{D}^+({\m{L}_{s}^{\text{sep}}})<\frac{1}{2}\,$. To do this, we will need a special case of the plane packing inequality by Folkman and Graham \cite{Folkman1969}:

\vspace*{3mm}
\emph{Let $K$ be a compact convex set. Then any subset of $K$ whose any two points are at least $1$ apart (in the Euclidean metric) has cardinality at most $\frac{2}{\sqrt{3}}\mrm{Area}(K)+\frac{1}{2}\mrm{Per}(K)+1$.}
\vspace*{3mm}

For our purposes $K$ will be a square of side length $R$, to be specified later. By scaling, we see that then any subset of $[0,R]^2$ whose any two points are at least $s$ apart has cardinality at most $\frac{2}{\sqrt{3}}\frac{R^2}{s^2}+\frac{2R}{s}+1$. Let $\theta\in \big(\frac{2}{\sqrt{3}}s^{-2},\frac{1}{2}\big)$, $\gamma\in (\sqrt{2},\theta^{-1/2})$, and $R>0$ such that $2(sR)^{-1}+R^{-2}\leq \theta-\frac{2}{\sqrt{3}} s^{-2}$.
Then, for all $\Lambda\in {\m{L}_{s}^{\text{sep}}}$, we have
\begin{equation*}
n^+(\Lambda,(0,R)^2)\leq \frac{2}{\sqrt{3}} \frac{R^2}{s^2}+\frac{2R}{s}+1\leq \left( \frac{2}{\sqrt{3}} s^{-2} + 2(s R)^{-1} +R^{-2} \right)R^2\leq \theta R^2,
\end{equation*}
and so, by Lemma \ref{lem:unif-lat-dist}, there exists an $R'=R'(\theta,\gamma,R)>0$ such that $\Lambda$ is $R'$-uniformly close to $\Omega_\gamma$. Therefore, as $\Lambda$ was arbitrary, it follows by Lemma \ref{lem:class-lat-dist} that $\m{D}^+({\m{L}_{s}^{\text{sep}}})\leq \gamma^{-2}<\frac{1}{2}$, and so $\msc{H}({\m{L}_{s}^{\text{sep}}})^p$ is identifiable by the probing signal $\varphi$.

Now consider the case $s\leq2\cdot 3^{-\frac{1}{4}}$ and let $r=2\cdot 3^{-\frac{1}{4}}$, $\Lambda'=\big\{r m(1,0)+r n\big(\frac{1}{2},\frac{\sqrt{3}}{2}\big): m,n\in\Z\big\}\in {\m{L}_{s}^{\text{sep}}}$.  Then $D^+(\Lambda')=\left(\frac{\sqrt{3}}{2}r^2\right)^{-1}=\frac{1}{2}$, and so by Lemma \ref{lem:Rayleigh-number} we have $\m{D}^+({\m{L}_{s}^{\text{sep}}})\geq D^+(\Lambda')=\frac{1}{2}$. It thus follows by Theorem \ref{NeceCritThm} that $\msc{H}({\m{L}_{s}^{\text{sep}}})^p$ is not identifiable.
\end{proof}

\begin{proof}[Proof of Corollary \ref{cor:RayClass}]
Consider first $\theta<\frac{1}{2}$. Fix $\gamma\in (\sqrt{2},\theta^{-1/2})$, and let $\Lambda\in {\m{L}_{s,\theta,R}^{\text{Ray}}}$. Then, by Lemma \ref{lem:unif-lat-dist}, there exists an $R'=R'(\theta,\gamma,R)>0$ such that $\Lambda$ is $R'$-uniformly close to $\Omega_\gamma$. Therefore, as $\Lambda$ was arbitrary, it follows by Lemma \ref{lem:class-lat-dist} that $\m{D}^+({\m{L}_{s,\theta,R}^{\text{Ray}}})\leq \gamma^{-2}<\frac{1}{2}$, and so Theorem \ref{CritThm} implies that $\msc{H}({\m{L}_{s,\theta,R}^{\text{Ray}}})^p$ is identifiable by the probing signal $\varphi$.

For $\theta> \frac{1}{2}$, suppose by way of contradiction that there exists a sequence $\{R_n\}_{n\in\N}$ of positive numbers such that $\lim_{n\to\infty}R_n=\infty$ and $\msc{H}({\m{L}_{s,\theta,R_n}^{\text{Ray}}})^p$ is  identifiable, for all $n\in\N$. Let $\{\gamma_n\}_{n\in\N}$ be a sequence of positive numbers such that $\gamma_n>\theta^{-1/2}>s$ and
\begin{equation*}
\gamma_n^{-2}+4\left(R_n^{-1}\gamma_n^{-1}+R_n^{-2}\right)\leq \theta,
\end{equation*}
for all $n\in\N$, and $\lim_{n\to\infty}\gamma_n=\theta^{-1/2}$. We then have
\begin{equation*}
n^+(\Omega_{\gamma_n},(0,R_n)^2)\leq \left(R_n\gamma_n^{-1}\right)^{2}+4\left(R_n\gamma_n^{-1}+1\right)\leq\theta R_n,
\end{equation*}
where the second term takes into account the points of $\Omega_{\gamma_n}$ along the four sides of the square $(0,R_n)^2$. Therefore $\Omega_{\gamma_n}\in \msc{H}({\m{L}_{s,\theta,R_n}^{\text{Ray}}})^p$, for all $n\in\N$, and so, as $\lim_{n\to\infty}\gamma_n^{-2}=\theta$ and $\theta>\frac{1}{2}$, it follows by Lemma \ref{lem:Rayleigh-number} that $\m{D}^+({\m{L}_{s,\theta,R_n}^{\text{Ray}}})\geq D^+(\Omega_{\gamma_n})=\gamma_n^{-2}>\frac{1}{2}$, for all sufficiently large $n$. This stands in contradiction to Theorem \ref{NeceCritThm}, completing the proof.
\end{proof}










\section*{Acknowledgments}
The authors would like to thank D. Stotz and R. Heckel for inspiring discussions and H. G. Feichtinger for his comments on an earlier version of the manuscript.

\bibliographystyle{IEEEtran} 
\bibliography{ref}

\section*{Appendix: proofs of auxiliary results}

\begin{proof}[Proof of Proposition \ref{wellDefProp}]
Note that $\linOp_\Lambda$ is, in fact, the synthesis operator treated in \cite{Groechenig2015}, and so items (i) and (ii) follow immediately from \cite[\S 2.5]{Groechenig2015}. We proceed to establish (iii). To this end, recall that $\varphi(t)=2^{\frac{1}{4}} e^{-\pi t^2}$, fix an arbitrary $(x,\omega)\in\Lambda$, and let $\wtd\Lambda$ be a finite subset of $\Lambda$ such that $(x,\omega)\in\wtd{\Lambda}$. Then, for $a,b\in\R$ and $(\tau,\nu)\in\wtd{\Lambda}$, we have
\begin{equation*}
\begin{aligned}
&\langle \modu_\nu \tra_\tau\, \modu_{-\omega-b}\tra_a\varphi , \tra_{x+a}\modu_{-b} \varphi \rangle_{M^q(\R)\times M^p(\R)}\\
=\,& \langle \varphi, \tra_{-a}\, \modu_{\omega+b}\tra_{-\tau} \modu_{-\nu} \tra_{x+a}\modu_{-b}\varphi \rangle_{L^2(\R)}\\
=\,& e^{-2\pi i (x+a)b}\langle \varphi, \tra_{-a}\, \modu_{\omega+b}\tra_{-\tau} \modu_{-\nu-b} \tra_{x+a}\varphi  \rangle_{L^2(\R)}\\
=\,& e^{-2\pi i (x+a)b}e^{2\pi i \tau(\nu+b)}\langle \varphi, \tra_{-a}\, \modu_{(\omega+b)+(-\nu-b)}\tra_{(x+a)-\tau}\varphi  \rangle_{L^2(\R)}\\
=\,& e^{-2\pi i (x+a)b}e^{2\pi i \tau(\nu+b)} e^{-2\pi i a (\omega-\nu)}\langle \varphi,\modu_{\omega-\nu}\tra_{(x+a)-\tau-a}\varphi  \rangle_{L^2(\R)}\\
=\,& e^{-2\pi i ab}e^{2\pi i \tau\nu} e^{-2\pi i a (\omega-\nu)-2\pi i b(x-\tau)} (\m{V}_{\varphi}\varphi)(x-\tau,\omega-\nu),
\end{aligned}
\end{equation*}
and therefore
\begin{equation}\label{eq:wellDefProp-1}
\begin{aligned}
&e^{2\pi i a b} \left\langle \linOp_{{\wtd{\Lambda}}} \,(\alpha, \modu_{-\omega-b}\tra_a\varphi) , \tra_{x+a}\modu_{-b} \varphi \right\rangle_{M^q(\R)\times M^p(\R)} \\
&\qquad\qquad\qquad=\sum_{(\tau,\nu)\in\wtd{\Lambda}}\alpha_{\tau,\nu}\, e^{2\pi i \tau\nu} e^{-2\pi i a (\omega-\nu)-2\pi i b(x-\tau)} (\m{V}_{\varphi}\varphi)(x-\tau,\omega-\nu).
\end{aligned}
\end{equation}
Now, let $\epsilon>0$ be arbitrary. We multiply both sides of \eqref{eq:wellDefProp-1} by $\epsilon \, e^{-\pi \epsilon (a^2+b^2)}$ and integrate over $(a,b)\in\R^2$. Then, as the Fourier transform of $\sqrt{\epsilon} \varphi(\cdot\,\sqrt{\epsilon})$ is $\varphi(\cdot /\sqrt{\epsilon})$, the integral of the right-hand side equals
\begin{equation*}
\sum_{(\tau,\nu)\in\wtd{\Lambda}}\alpha_{\tau,\nu}e^{2\pi i \tau\nu} e^{-\pi \left[ (\omega-\nu)^2+ (x-\tau)^2\right]\epsilon^{-1}} (\m{V}_{\varphi}\varphi)(x-\tau,\omega-\nu),
\end{equation*}
and the integral of the left-hand side satisfies
\begin{equation*}
\begin{aligned}
& \left| \int_{\R^2} \epsilon e^{-\pi \epsilon (a^2+b^2)}  e^{2\pi i a b} \left\langle \linOp_{{\wtd{\Lambda}}} \,(\alpha, \modu_{-\omega-b}\tra_a\varphi) , \tra_{x+a}\modu_{-b} \varphi \right\rangle_{M^q(\R)\times M^p(\R)} \, \mrm{d}a \, \mrm{d}b \;\right|\\
\leq\,& \int_{\R^2} \epsilon e^{-\pi \epsilon (a^2+b^2)} \left|\left\langle \linOp_{{\wtd{\Lambda}}} \,(\alpha, \modu_{-\omega-b}\tra_a\varphi) , \tra_{x+a}\modu_{-b} \varphi \right\rangle_{M^q(\R)\times M^p(\R)}\right| \, \mrm{d}a \, \mrm{d}b \;\\
\leq\,&  \int_{\R^2} \epsilon e^{-\pi \epsilon (a^2+b^2)}\|\linOp_{\wtd{\Lambda}}(\alpha,\cdot)\| \| \modu_{-\omega-b}\tra_a\varphi\|_{M^1(\R)} \| \tra_{x+a}\modu_{-b} \varphi\|_{M^q(\R)} \, \mrm{d}a \, \mrm{d}b \\
\leq\, & \|\linOp_{\wtd{\Lambda}}(\alpha,\cdot)\| \|\varphi\|_{M^1(\R)} \|\varphi\|_{M^q(\R)}.
\end{aligned}
\end{equation*}
We hence deduce that
\begin{equation*}
\left|\sum_{(\tau,\nu)\in\wtd{\Lambda}}\alpha_{\tau,\nu}e^{2\pi i \tau\nu} e^{-\pi \left[ (\omega-\nu)^2+ (x-\tau)^2\right]\epsilon^{-1}} (\m{V}_{\varphi}\varphi)(x-\tau,\omega-\nu)\right| \leq  \|\linOp_{\wtd{\Lambda}}(\alpha,\cdot)\| \|\varphi\|_{M^1(\R)} \|\varphi\|_{M^q(\R)},
\end{equation*}
and so, upon letting $\epsilon\to 0$, we obtain
\begin{equation*}
\left| \alpha_{x,\omega}e^{2\pi i x\omega}  (\m{V}_{\varphi}\varphi)(0,0)\right|\leq \|\linOp_{\wtd{\Lambda}}(\alpha,\cdot) \| \|\varphi\|_{M^1(\R)} \|\varphi\|_{M^q(\R)}.
\end{equation*}
Next, note that we can write $\linOp_{{\wtd{\Lambda}}}=\linOp_{\Lambda}-\linOp_{{\Lambda\setminus \wtd{\Lambda}}}$, and so, by item (i) of the proposition, there exists a universal constant $C>0$ such that
\begin{equation*}
\|\linOp_{{\wtd{\Lambda}}}(\alpha,\cdot )\|\leq \|\linOp_\Lambda(\alpha,\cdot )\| + C\, \|\{\alpha_\lambda\}_{\lambda\in \Lambda\setminus \wtd\Lambda}\|_{\ell^p}.
\end{equation*}
Therefore, since $(\m{V}_{\varphi}\varphi)(0,0)=\|\varphi\|_{L^2(\R)}^2=1$, we have
\begin{equation*}
|\alpha_{x,\omega}|\leq \left( \|\linOp_{\Lambda}(\alpha,\cdot )\| + C \, \|\{\alpha_\lambda\}_{\lambda\in \Lambda\setminus \wtd\Lambda}\|_{\ell^p}\right) \|\varphi\|_{M^1(\R)} \|\varphi\|_{M^q(\R)}.
\end{equation*}
The term $\|\{\alpha_\lambda\}_{\lambda\in \Lambda\setminus \wtd\Lambda}\|_{\ell^p}$ can now be made arbitrarily small by choosing a sufficiently large $\wtd{\Lambda}$, and hence we deduce that
\begin{equation*}
|\alpha_{x,\omega}|\leq  \|\linOp_{\Lambda}(\alpha,\cdot )\| \|\varphi\|_{M^1(\R)} \|\varphi\|_{M^q(\R)}.
\end{equation*}
As $(x,\omega)\in\Lambda$ was arbitrary, we obtain
\begin{equation*}
\|\alpha\|_{\ell^\infty}= \sup_{(x,\omega)\in\Lambda}|\alpha_{x,\omega}|\leq \|\linOp_{\Lambda}(\alpha,\cdot )\| \|\varphi\|_{M^1(\R)} \|\varphi\|_{M^q(\R)},
\end{equation*}
establishing \eqref{eq:wellDefProp-0}.
\end{proof}

\begin{proof}[Proof of Lemma \ref{lem:Rayleigh-number}]
Item (i) is a direct consequence of Definition \ref{def:UBcd} and the fact that $n^+(\Lambda, (0,R)^2)\allowbreak \leq n^+(\Lambda, [0,R]^2) $. To show item (ii), let $\epsilon>0$ be arbitrary, and set $\theta=\m{D}^+(\m{L})+\epsilon $. Then, by Definition \ref{def:UBcd}, we have
\begin{equation*}
\frac{n^+(\Lambda, [0,R]^2)}{R^2}\leq \m{D}^+(\m{L})+\epsilon
\end{equation*}
for every $\Lambda\in \m{L}$ and sufficiently large $R$, and thus
\begin{equation}\label{eq:Helmut_1}
D^{+}(\Lambda) =  \limsup_{R \to \infty} \frac{n^+(\Lambda, [0,R]^2)}{R^2}\leq  \m{D}^+(\m{L})+\epsilon.
\end{equation}
Now, as \eqref{eq:Helmut_1} holds for every $\Lambda\in\m{L}$, we deduce that $\sup_{\Lambda \in \m{L}} D^+(\Lambda)\leq \m{D}^+(\m{L})+\epsilon$, and hence, as $\epsilon>0$ was arbitrary, we obtain $\sup_{\Lambda \in \m{L}} D^+(\Lambda)\leq \m{D}^+(\m{L})$.
\end{proof}

\begin{proof}[Proof of Lemma \ref{lem:unif-lat-dist}]
We identify $\C$ with $\R^2$ for ease of exposition. For a positive integer $q$, define the set $\m{S}_q$ according to
\begin{equation*}
\m{S}_q=\{[q k R, q(k+1)R]\times [q\ell R,q(\ell +1)R] \, :\, k,\ell\in\Z \},
\end{equation*}
and note that $\m{S}_q$ is a collection of squares in $\R^2$ of side length $qR$ tessellating the plane.
A simple counting argument now yields,  for all $K\in\m{S}_q$,
\begin{equation*}
\#(\Omega_\gamma\cap K)\geq (qR\gamma^{-1})^2 - 4\left( qR\gamma^{-1}+1\right),
\end{equation*}
where the subtracted term accounts for the $\lceil qR\gamma^{-1}\rceil$ points adjacent to each of the edges of the square $K$, but which might not be inside it. On the other hand, as every element of $\m{S}_q$ can be covered by $(q+1)^2$ translates of $(0,R)^2$, using $n^+(\Lambda , (0,R)^2)\leq \theta R^2$ we have $\#(\Lambda \cap K) \leq (q+1)^2 \cdot \theta R^2 =((q+1)R \theta^{1/2})^2$ for all $K\in\m{S}_q$.  Therefore, as $\gamma^{-1}> \theta^{1/2}$ by assumption, there exists a positive integer $q'=q'(\theta,R,\gamma)$ such that
\begin{equation*}
 (q' R\gamma^{-1})^2 - 4\left( q' R\gamma^{-1}+1\right)\geq \big((q'+1) R \theta^{1/2} \big)^2.
\end{equation*}
We thus have $\#(\Omega_\gamma\cap K)\geq \#(\Lambda \cap K) $, for all $K\in\m{S}_{q'}$, and can therefore enumerate $\Lambda=\{\lambda_{m,n}\}_{(m,n)\in\m{I}}$ so that, for every $(m,n)\in \m{I}$, $\lambda_{m,n}$ and  $\omega_{m,n}$ are contained in the same square $K_{m,n}\in \m{S}_{q'}$. Setting $R'=\sqrt{2}q' R$ to be the length of the diagonal of the squares in $\m{S}_{q'}$ now yields $|\lambda_{m,n}-\omega_{m,n}|\leq R'$, for all $(m,n)\in\m{I}$, as desired.
\end{proof}

\begin{proof}[Proof of Lemma \ref{lem:class-lat-dist}]
We again identify $\C$ with $\R^2$ for ease of exposition.
Note that we have $n^+(\Lambda,[0,R']^2)\leq n^+(\Omega_\gamma ,[0,R'+2R]^2)$, for all $R'>0$ and $\Lambda\in\m{L}$, by the uniform closeness assumption, and so
\begin{equation*}
\m{D}^+(\m{L})=\limsup_{R'\to\infty}\sup_{\Lambda\in\m{L}}\frac{n^+(\Lambda,[0,R']^2)}{R'^2}\leq \limsup_{R'\to\infty}\frac{n^+(\Omega_\gamma,[0,R'+2R]^2)}{(R'+2R)^2}\left(1+\frac{2R}{R'}\right)^2= \gamma^{-2}.
\end{equation*}
\end{proof}

\begin{proof}[Proof of Lemma \ref{massiveLemma}]
Note that the terms $z-\lambda_{0,0}$ and $z$ in the defining expressions of $g_\Lambda$ and $\wtd{g}_\Lambda$ cancel owing to $\lambda_{0,0}=0$, so the interpolation property (a) follows immediately. We may hence proceed to establishing statement (b). 
To this end, we begin by observing that the assumptions (i), (ii), and (iii) remain valid and the conclusion of the lemma unchanged if we replace $\rho$ by $\rho\wedge 1$ and $R$ by $R\vee 1\vee (3\gamma)$, so we may assume w.l.o.g. that $\rho\leq 1$ and $R\geq 1\vee (3\gamma)$.
Now, set $\m{I}' \coleqq \m{I}\setminus\{(0,0)\}$ and write
\begin{equation}\label{longLem-1stEq}
\widetilde{g}_\Lambda(z)e^{-\frac{\pi}{2}\gamma^{-2}|z|^2}=\frac{\sigma_\gamma(z)e^{-\frac{\pi}{2}\gamma^{-2}|z|^2}}{d(z,\Omega_\gamma)}\,\frac{d(z,\Lambda)}{z}\,h(z),
\end{equation}
where $d$ is the Euclidean distance from a point to a subset of $\C$ and
\begin{equation*}
h(z)=\frac{d(z,\Omega_\gamma)}{d(z,\Lambda)} \prod_{(m,n)\in\m{I}_s }\exp\left(\frac{z}{\omega_{m,n}}-\frac{z}{\lambda_{m,n}}\right) \prod_{(m,n)\in\m{I}'}\frac{(1-z/\lambda_{m,n})\exp(z/\lambda_{m,n})}{(1-z/\omega_{m,n})\exp(z/\omega_{m,n})}\notag.
\end{equation*}
Note that $|{d(z,\Lambda)}/{z}|\leq 1$ owing to $0\in\Lambda$, and by Lemma \ref{vanillaSigmaBd} we have $|\sigma_\gamma(z)|e^{-\frac{\pi}{2}\gamma^{-2}|z|^2}/d(z,\Omega_\gamma)\lesssim_{\,\gamma}\!1$, so in order to complete the proof it suffices to establish
\begin{equation}\label{eq:h-goal-bound}
|h(z)|\lesssim_{\,s,\theta,\gamma,R} \rho^{-1} e^{c |z|\log{|z|}},\quad\forall z\in\C,
\end{equation}
for some $c=c(s,\theta,\gamma,R)>0$. This will be effected by bounding various auxiliary quantities associated with $h$. 
To this end, we begin by bounding
\begin{equation*}
h_{\mrm{aux}}(z)\coleqq \prod_{(m,n)\in A(z)}\frac{(1-z/\lambda_{m,n})\exp(z/\lambda_{m,n})}{(1-z/\omega_{m,n})\exp(z/\omega_{m,n})},
\end{equation*}
from above, where $A(z)\subset\m{I}$ is given by
\begin{equation*}
A(z)=\left\{(m,n)\in\m{I}': |\lambda_{m,n}|>2|z|\vee 6R \right\}.
\end{equation*}
To this end, we first establish the following basic bounds valid for all $z\in\C$ and $(m,n)\in A(z)$:
\begin{enumerate}[leftmargin=12mm, label=(A\arabic*)]
\item $|\omega_{m,n}|>5R$, \; ${|\omega_{m,n}|/|\lambda_{m,n}|}\in\left({ 5/6, 5/4}\right)$,  
\item  $|z/\omega_{m,n}|<\frac{4}{5} \mathds{1}_{\{ |z|\leq 4R\} }  + \frac{4}{7} \mathds{1}_{\{ |z|>4R\}}$, and
\item $|z/\lambda_{m,n}|<1/2$.
\end{enumerate}
The inequality in (A1) follows from $|\omega_{m,n}|\geq |\lambda_{m,n}|-R>6R-R=5R$, and the upper and lower bounds on $|\omega_{m,n}|/|\lambda_{m,n}|$ are due to $|\omega_{m,n}|\geq |\lambda_{m,n}|-R>\frac{5}{6}|\lambda_{m,n}|$ and $|\lambda_{m,n}|\geq |\omega_{m,n}|-R>\frac{4}{5}|\omega_{m,n}|$. To show (A2), consider the cases $|z|\leq 4R$ and $|z|>4R$ separately. If $|z|\leq 4R$, then $|z/\omega_{m,n}|<4R/5R=4/5$, whereas if $|z|>4R$, then $|\omega_{m,n}|\geq |\lambda_{m,n}|-R>2|z|-R>\frac{7}{4}|z|$ and so $|z/\omega_{m,n}|<4/7$. Finally, (A3) follows directly by the definition of $A(z)$.

Now, using (A2), we have
\begin{align}
|h_{\mrm{aux}}(z)|&\leq \prod_{(m,n)\in A(z)}\frac{|1-z/\lambda_{m,n}| |\exp(z/\lambda_{m,n})|}{|1-z/\omega_{m,n}| |\exp(z/\omega_{m,n})|} \notag\\
&\leq\Bigg(\mathds{1}_{\{|z|> 4R\}} \;  + \mathds{1}_{\{|z|\leq 4R\}}\hspace*{-4mm} \prod_{\substack{(m,n)\in A(z)\\ |\omega_{m,n}|\leq \frac{7}{4}|z|\leq 7R }} \hspace*{-3mm} \frac{\left(1+\frac{1}{2}\right)e^{\frac{1}{2}}}{\left(1-\frac{4}{5}\right)e^{-\frac{4}{5}}} \Bigg)\prod_{\substack{(m,n)\in A(z)\\ |\omega_{m,n}|> \frac{7}{4}|z| }}\frac{|1-z/\lambda_{m,n}| |\exp(z/\lambda_{m,n})|}{|1-z/\omega_{m,n}| |\exp(z/\omega_{m,n})|} \notag\\
&\lesssim_{\,\gamma,R} \Bigg| \prod_{(m,n)\in \wtd A(z)}\frac{(1-z/\lambda_{m,n})\exp(z/\lambda_{m,n})}{(1-z/\omega_{m,n})\exp(z/\omega_{m,n})}\Bigg|,\quad\forall z\in\C \label{case<4Q-h4est-pre},
\end{align}
where we set
\begin{equation*}
\wtd{A}(z)=\{(m,n)\in A(z) : |\omega_{m,n}|> {\textstyle \frac{7}{4}}|z|\}.
\end{equation*}
Next, for $z\in\C$ and $(m,n)\in A(z)$, it follows by  (A2) that $1-{z/\omega_{m,n}}$ and $1-{z/z_{m,n}}$ lie in the domain $\C\setminus\R_{\leq 0}$ of the complex logarithm, which we denote by $\mrm{Log}$. We can thus write
\begin{align}
& \prod_{(m,n)\in \wtd A(z)}\frac{(1-z/\lambda_{m,n})\exp(z/\lambda_{m,n})}{(1-z/\omega_{m,n})\exp(z/\omega_{m,n})}\notag\\
=\, &\prod_{(m,n)\in \wtd A(z)}\exp{\left[\mrm{Log}\left(1-\frac{z}{\lambda_{m,n}}\right)+\frac{z}{\lambda_{m,n}}-\mrm{Log}\left(1-\frac{z}{\omega_{m,n}}\right)-\frac{z}{\omega_{m,n}}\right]}\notag\\
= \,&\prod_{(m,n)\in \wtd A(z)}\exp{\left[-\sum_{k=2}^{\infty}\frac{1}{k}\left(\frac{z}{\lambda_{m,n}}\right)^k+\sum_{k=2}^{\infty}\frac{1}{k}\left(\frac{z}{\omega_{m,n}}\right)^k\right]}\notag\\
=\,&\exp{\Bigg[\sum_{(m,n)\in \wtd A(z)} \sum_{k=2}^\infty \frac{z^k}{k}\left(\frac{1}{\omega_{m,n}^k}-\frac{1}{\lambda_{m,n}^k}\right)\Bigg]},\notag
\end{align}
which together with \eqref{case<4Q-h4est-pre} yields
\begin{equation}\label{case<4Q-h4est}
|h_{\mrm{aux}}(z)|\lesssim_{\,\gamma,R} \exp\Bigg[\sum_{(m,n)\in \wtd A(z)} \sum_{k=2}^\infty \frac{|z|^k}{k}\abs{\frac{1}{\omega_{m,n}^k}-\frac{1}{\lambda_{m,n}^k}} \Bigg],\quad \forall z\in\C. 
\end{equation}
Using (A1), we further have
\begin{align}
\abs{\frac{1}{\omega_{m,n}^k}-\frac{1}{\lambda_{m,n}^k}}&=\frac{\abs{\lambda_{m,n}-\omega_{m,n}}\sum_{\ell=0}^{k-1}|\lambda_{m,n}|^{k-1-\ell}|\omega_{m,n}|^\ell}{\abs{\omega_{m,n}}^k\abs{\lambda_{m,n}}^k}\notag\\
&\leq\frac{R\cdot |\omega_{m,n}|^{k-1} \sum_{\ell=0}^{k-1}\left(\frac{6}{5}\right)^\ell}{\abs{\omega_{m,n}}^{2k}\left(\frac{4}{5}\right)^k}\notag\\
&\leq 5R \cdot { 1.5 }^{k}\abs{\omega_{m,n}}^{-(k+1)}, \label{eq:mass-lem-recip-est}
\end{align}
for all $z\in\C$, $(m,n)\in A(z)$, and all $k\in \N$.
Next, defining $c_{\mrm{aux}}(z,R)\coleqq \frac{7|z|}{4}\vee 5R$, (A1) and (A2) imply that
$\wtd{A} (z)\subset\{(m,n)\in\Z^2 : |\omega_{m,n}|> c_{\mrm{aux}}(z,R)\}$, for all $z\in\C$, and so \eqref{case<4Q-h4est} and \eqref{eq:mass-lem-recip-est} together yield
\begin{align}
|h_{\mrm{aux}}(z)|&\leq \exp\Bigg[ \sum_{\substack{(m,n)\in \wtd{A}(z) }}\sum_{k=2}^\infty{  \frac{5R |z|^k}{k}}\cdot 1.5^k |\omega_{m,n}|^{-(k+1)} \Bigg]\notag\\
&\leq  \exp\Bigg[\sum_{\substack{(m,n)\in\Z^2 \\ |\omega_{m,n}|>c_{\mrm{aux}}(z,R)}}\quad \sum_{k=2}^\infty {  \frac{5R |z|^k}{k}} \cdot 1.5^k |\omega_{m,n}|^{-(k+1)} \Bigg]\notag\\
&=\exp\Bigg[ \sum_{k=2}^\infty { \frac{5R  |z|^k}{k}} \cdot 1.5^k \cdot \sum_{\substack{(m,n)\in\Z^2 \\ |\omega_{m,n}|>c_{\mrm{aux}}(z,R)}} |\omega_{m,n}|^{-(k+1)} \Bigg]. \label{use-int-est-1}
\end{align}
Recall that $R\geq 3\gamma$, and so $c_{\mrm{aux}}(z,R)=\frac{7|z|}{4}\vee 5R\geq \frac{7|z|}{4} \vee \frac{20\gamma }{\sqrt{2}}$ implies
\begin{equation}\label{eq:int-test-weird-c1-bd}
c_{\mrm{aux}}(z,R)-\gamma/\sqrt{2}\geq  0.95\, c_{\mrm{aux}}(z,R) \geq 1.6\, |z|.
\end{equation}
Now, write $K_\gamma$ for the square of side length $\gamma$ centered at $(0,0)$ and use $R\geq 3\gamma$ and \eqref{eq:int-test-weird-c1-bd} to obtain the following bound
\begin{align}
\sum_{|\omega_{m,n}|>c_{\mrm{aux}}(z,R)}|\omega_{m,n}|^{-(k+1)}&\leq \sum_{|\omega_{m,n}|>c_{\mrm{aux}}(z,R)}\left(\frac{|\omega_{m,n}|+\frac{\gamma}{\sqrt{2}}}{|\omega_{m,n}|}\right)^{k+1} \int_{(m\gamma,n\gamma)+K_\gamma} |w|^{-(k+1)}\abs{\mrm{d}w}\\
&\leq \left(1+\frac{\gamma}{5R\sqrt{2}}\right)^{k+1}\int_{|w|\geq c_{\mrm{aux}}(z,R)-\frac{\gamma}{\sqrt{2}}}|w|^{-(k+1)}\abs{\mrm{d}w}\notag\\
&\leq \left(1+ \frac{1}{15\sqrt{2}}\right)^{k+1} \frac{2\pi}{k-1}\left(c_{\mrm{aux}}(z,R)-\gamma/\sqrt{2}\right)^{1-k}\notag\\
&\leq  1.04^{k+1} \cdot 2\pi \left(1.6 |z|\right)^{1-k}\notag\\
&< 11 \cdot 0.65^{k} |z|^{1-k},\label{int-est-1}
\end{align}
for all $z\in\C\setminus\{0\}$ and $k\geq 2$. We now use \eqref{int-est-1} in \eqref{use-int-est-1} to obtain
\begin{align}
|h_{\mrm{aux}}(z)|&\lesssim_{\,\gamma,R} \exp \Bigg[ \sum_{k=2}^\infty \frac{5 R|z|^k}{k} \cdot 1.5^k \cdot 11 \cdot 0.65^{k} |z|^{1-k}\Bigg]\notag \\
&\leq {\Bigg[  R\,|z|\cdot \frac{55}{2} \sum_{k=2}^\infty 0.975^k\Bigg]}\notag\\
&=e^{1100 R\, |z|},\quad\forall z\in\C.\label{eq:haux-final-bd-1}
\end{align}

We are now ready to bound $h$, and we do so by treating  the cases $|z|> 3R$ and $|z|\leq3R$ separately.

\noindent \textit{Case $|z|>3R\,$}:
We analyze $h(z)$ as a product
$
h(z)=\prod_{k=1}^{5}h_k(z)
$, where
\begin{align}
h_1(z)&=\exp\Bigg[ z\sum_{\substack{(m,n)\in\m{I}'\setminus\m{I}_s \\ |\lambda_{m,n}|\leq 2|z|}}\left(\frac{1}{\lambda_{m,n}}-\frac{1}{\omega_{m,n}}\right)\Bigg], & h_2(z)&=\frac{d(z,\Omega_\gamma)}{d(z,\Lambda)}\prod_{\substack{(m,n)\in\m{I}' \\ |\lambda_{m,n}-z|\leq 2R}} \frac{1-z/\lambda_{m,n}}{1-z/\omega_{m,n}}, \notag \\[5mm]
h_3(z)&=\prod_{\substack{(m,n)\in \m{I}' \\ |\lambda_{m,n}-z|>2R \\ |\omega_{m,n}|\leq 2R }}\frac{1-z/\lambda_{m,n}}{1-z/\omega_{m,n}},
&  h_4(z)&= \prod_{\substack{(m,n)\in\m{I}' \\ |\lambda_{m,n}-z|>2R \\|\omega_{m,n}|>2R\\|\lambda_{m,n}|\leq 2|z|}}\frac{1-z/\lambda_{m,n}}{1-z/\omega_{m,n}},\quad\text{and}\notag\\
h_5(z)&=\prod_{\substack{(m,n)\in\m{I}'\\ |\lambda_{m,n}|>2|z|}}\frac{(1-z/\lambda_{m,n})\exp(z/\lambda_{m,n})}{(1-z/\omega_{m,n})\exp(z/\omega_{m,n})}, & & \notag
\end{align}
and bound the functions $h_j$ in order.

\noindent \textit{Bounding $h_1$}:
Note that $|\lambda_{m,n}|>3R$ implies $|\omega_{m,n}|>2R$, and hence $|\lambda_{m,n}|\geq |\omega_{m,n}|-R>\frac{1}{2}|\omega_{m,n}|$. We thus have the following bound for $(m,n)\in \m{I}'\setminus\m{I}_s$:
\begin{equation*}
\begin{aligned}
\left|\frac{1}{\lambda_{m,n}}-\frac{1}{\omega_{m,n}}\right|&\leq \mathds{1}_{\{|\lambda_{m,n}|\leq 3R\}}(2s^{-1}+ \gamma^{-1})+ \mathds{1}_{\{|\lambda_{m,n}|> 3R\}}\frac{|\lambda_{m,n}-\omega_{m,n}|}{|\omega_{m,n}||\lambda_{m,n}|}\\
&\leq  \mathds{1}_{\{|\lambda_{m,n}|\leq 3R\}}(2s^{-1}+\gamma^{-1})+2R|\omega_{m,n}|^{-2},
\end{aligned}
\end{equation*}
and therefore, as $\#\{(m,n)\in \m{I}': |\lambda_{m,n}|\leq 3R\}\leq 9\theta R^2$,
\begin{align}
|h_1(z)|&\leq \exp \Bigg[ |z| \hspace*{-0.5em}\sum_{\substack{(m,n)\in\m{I}'\\|z_{m,n}|\leq 2|z| }}\left|\frac{1}{z_{m,n}}-\frac{1}{\omega_{m,n}}\right|\Bigg],\notag\\
&\leq  \exp \Bigg[9\theta R^2(2s^{-1}+\gamma^{-1})|z| +2R|z| \cdot \hspace*{-1em}\sum_{\substack{(m,n)\in\m{I}' \\|\omega_{m,n}|\leq 2|z|+R}}\abs{\omega_{m,n}}^{-2}\Bigg].\label{h1-first-est}
\end{align}
Recalling that $\log|z|>\log(3R)>\log (3)>0$, we can bound in a manner similar to \eqref{int-est-1} to obtain
\begin{align}
\sum_{\substack{(m,n)\in\m{I}' \\|\omega_{m,n}|\leq 2|z|+R}}\abs{\omega_{m,n}}^{-2}&\leq \left(\frac{\gamma+\frac{\gamma}{\sqrt{2}}}{\gamma}\right)^{2}\int_{\frac{\gamma}{\sqrt{2}}\leq|w|\leq 2|z|+R+\frac{\gamma}{\sqrt{2}}}|w|^{-2}\abs{\mrm{d}w}\notag\\
&\lesssim_{\, \gamma,R}\log |z| .
\end{align}
Using this in \eqref{h1-first-est} thus yields
\begin{equation}
|h_1(z)|\leq e^{ c_1 |z|\log{|z|}},\label{h1-final-est}
\end{equation}
for some $c_1=c_1(s,\theta,\gamma,R)>0$ and all $z\in\C$ such that $|z|> 3R$.

\noindent \textit{Bounding $h_2$}: We write $|h_2|=h_2^{\mrm{num}}/h_{2}^{\mrm{den}}$, where
\begin{equation*}
h_2^{\mrm{num}}(z)=\frac{1}{d(z,\Lambda)}\hspace*{-0.5em}\prod_{\substack{ (m,n)\in\m{I}'\\ |\lambda_{m,n}-z|\leq 2R }}\hspace*{-3mm} \frac{|\lambda_{m,n}-z|}{|\lambda_{m,n}|},\quad\text{and}\quad
h_2^{\mrm{den}}(z)=\frac{1}{d(z,\Omega_\gamma)} \prod_{\substack{ (m,n)\in\m{I}' \\ |\lambda_{m,n}-z|\leq 2R }}\hspace*{-3mm}\frac{|\omega_{m,n}-z|}{|\omega_{m,n}|}.
\end{equation*}
In order to bound $h_2^{\mrm{num}}$, we observe that one of the following two circumstances occurs:
\begin{itemize}
\item[--] The distance from $z$ to $\Lambda$ is minimized at $\lambda_{0,0}$, and so $d(z,\Lambda)=|z|>3R>1$.
\item[--] The distance from $z$ to $\Lambda$ is minimized at a point $\lambda_{m,n}$, where $(m,n)\in\m{I}'$, and so the term $d(z,\Lambda)$ cancels with one of the factors $|\lambda_{m,n}-z|$ in the product over $\{(m,n)\in\m{I}': |\lambda_{m,n}-z|\leq 2R\}$.
\end{itemize}
These facts lead to the following bound
\begin{equation}\label{eq:h2nom-bd}
h_2^{\mrm{num}}(z)\leq \prod_{\substack{ (m,n)\in\m{I}'\\ |\lambda_{m,n}-z|\leq 2R}}\frac{|\lambda_{m,n}-z|\vee 1}{|z|-|\lambda_{m,n}-z|}\leq 2^{16\theta R^2}.
\end{equation}

For $h_2^{\mrm{den}}$, we similarly observe that either $(\mrm{Re}\, z,\mrm{Im}\, z)\in[-\frac{\gamma}{2},\frac{\gamma}{2}]^2$, or $d(z,\Omega_\gamma)$ cancels with a factor $|\omega_{m,n}-z|$ in the product over $\{(m,n)\in\m{I}': |\lambda_{m,n}-z|\leq 2R\}$. In either case the numerators of the terms remaining in the product satisfy $|\omega_{m,n}-z|\geq\frac{\gamma}{2}$, and we thus have 
\begin{equation}\label{eq:h2den-bd}
h_2^{\mrm{den}}(z)\geq \left(\!\frac{\gamma}{\sqrt{2}}\wedge{1}\!\right)\hspace*{-3mm} \prod_{\substack{ (m,n)\in\m{I}'\\ |\lambda_{m,n}-z|\leq 2R }}\frac{(\gamma/2)\wedge 1}{|\omega_{m,n}-\lambda_{m,n}|+|\lambda_{m,n}-z|+|z|}\geq  \left(\!\frac{\gamma}{\sqrt{2}}\wedge 1 \!\right) \left(\frac{(\gamma /2)\wedge 1}{3R+|z|}\right)^{16\theta R^2}.
\end{equation}
The inequalities \eqref{eq:h2nom-bd} and \eqref{eq:h2den-bd} together yield
\begin{equation}
|h_2(z)|\lesssim_{\,\theta,\gamma,R} |z|^{16\theta R^2}\lesssim_{\,\theta,R}e^{|z|}.\label{eq:h2-final-est}
\end{equation}

\noindent \textit{Bounding $h_3$}:  Recall that $|\lambda_{m,n}|> \frac{s}{2}$ for all $(m,n)\in\m{I}'\setminus\m{I}_s $, and there is at most one $(m',n')\in \m{I}_s\setminus\{(0,0)\}$, and for this $(m',n')$ we have $|\lambda_{m',n'}|\geq \rho$, by assumption (i). We thus get
\begin{align}
|h_3(z)|&=\prod_{\substack{(m,n)\in\m{I}'\\ |\lambda_{m,n}-z|>2R\\ |\omega_{m,n}|\leq 2R}}\frac{|1-z/\lambda_{m,n}|}{|1-z/\omega_{m,n}|}\leq\prod_{\substack{(m,n)\in\m{I}'\\ |\lambda_{m,n}-z|>2R\\ |\omega_{m,n}|\leq 2R}}\frac{|\omega_{m,n}|}{|\lambda_{m,n}|}\frac{|\lambda_{m,n}-z|}{|\lambda_{m,n}-z|-R}\notag\\
&\leq \rho^{-1}\prod_{|\omega_{m,n}|\leq 2R} \frac {2R}{(s/2)\wedge 1}\,\frac{2R}{2R-R}\leq \rho^{-1} \left(\frac{4R}{(s/2)\wedge 1}\right)^{16\gamma^{-2}R^2}\notag\\
&\lesssim_{\,s,\gamma,R } \rho^{-1}.\label{eq:h3-only-est}
\end{align}

\noindent \textit{Bounding $h_4$}:  We write $|h_4|=h_4^{\mrm{num}}/h_4^{\mrm{den}}$, where
\begin{equation*}
h_4^{\mrm{num}}(z)=\prod_{\substack{(m,n)\in\m{I}' \\ |\lambda_{m,n}-z|>2R \\|\omega_{m,n}|>2R \\|\lambda_{m,n}|\leq 2|z|}}\left|1-\frac{\omega_{m,n}-\lambda_{m,n}}{\omega_{m,n}-z}\right|,\quad\text{and}\quad h_4^{\mrm{den}}(z)=\prod_{\substack{(m,n)\in\m{I}' \\ |\lambda_{m,n}-z|>2R \\|\omega_{m,n}|>2R \\|\lambda_{m,n}|\leq 2|z|}} \left|1-\frac{\omega_{m,n}-\lambda_{m,n}}{\omega_{m,n}}\right|.
\end{equation*}
Now, fix a $z\in\C$ with $|z|>3R$ and write $z=z'+\omega_{k,\ell}$, where $(k,\ell)\in\Z^2$ and $(\mrm{Re}\, z',\mrm{Im}\, z')\in[-\frac{\gamma}{2},\frac{\gamma}{2}]^2$. Then, as
\begin{equation*}
 \{(m,n)\in\m{I}': |\lambda_{m,n}-z|\!>\!2R, |\lambda_{m,n}|\leq 2|z| \}\subset\{(m,n)\in\Z^2: |\omega_{m,n}-z|\!>\!R, |\omega_{m,n}-z|\leq 3|z|+R \},
\end{equation*}
and $\Omega_{\gamma}-\omega_{k,\ell}=\Omega_{\gamma}$, we have the following bound
\begin{align}
|h_4^{\mrm{num}}(z)|&\leq \prod_{\substack{(m,n)\in\m{I}' \\ |\lambda_{m,n}-z|>2R \\|\omega_{m,n}|>2R \\|\lambda_{m,n}|\leq 2|z|}} \left(1+\frac{|\omega_{m,n}-\lambda_{m,n}|}{|\omega_{m,n}-z|}\right)\leq  \prod_{\substack{(m,n)\in\Z^2 \\ R<|\omega_{m,n}-z|\leq 3|z|+R}} \left(1+\frac{R}{|\omega_{m,n}-z|}\right) \notag \\
&= \prod_{\substack{(m,n)\in\Z^2 \\ R<|\omega_{m,n}-\omega_{k,\ell} -z'|\leq 3|z|+R}} \left(1+\frac{R}{|\omega_{m,n}-\omega_{k,\ell }- z'|}\right) \notag\\ 
& \leq \exp \Bigg[\sum_{\substack{(m,n)\in\Z^2 \\ R<|\omega_{m,n} -z'|\leq 3|z|+R}} \frac{R}{|\omega_{m,n}- z'|} \Bigg],\label{h4-est-1st}
\end{align}
where in the last inequality we used $\log(1+x)\leq x$ for $x\geq 0$. Now, $|\omega_{m,n}-z'|> R$ and $R\geq 3\gamma$ imply $|\omega_{m,n}|>R-\frac{\gamma}{\sqrt{2}}>\sqrt{2}\gamma$ and so $|\omega_{m,n}-z'| \geq |\omega_{m,n}|-\frac{\gamma}{\sqrt{2}}>\frac{1}{2}|\omega_{m,n}|$. We therefore have
\begin{align}
\sum_{R<|\omega_{m,n}-z'|\leq 3|z|+R}\frac{1}{|\omega_{m,n}-z'|}&\leq \sum_{R-\frac{\gamma}{\sqrt{2}}<|\omega_{m,n}|\leq 3|z|+R+\frac{\gamma}{\sqrt{2}}}\frac{2}{|\omega_{m,n}|}\notag\\
&\leq \sum_{R-\frac{\gamma}{\sqrt{2}}<|\omega_{m,n}|\leq 3|z|+R+\frac{\gamma}{\sqrt{2}}} 2\frac{|\omega_{m,n}|+\frac{\gamma}{\sqrt{2}}}{|\omega_{m,n}|}\int_{(m\gamma,n\gamma)+K_\gamma}|w|^{-1}|\mrm{d}w|\notag\\
&\lesssim_{\,\gamma,R} \frac{R}{R-\frac{\sqrt{2}}{2}}\int_{R-\sqrt{2}\gamma<|\omega|\leq 3|z|+R+\sqrt{2}\gamma}|w|^{-1}|\mrm{d}w|\notag\\
&\lesssim_{\,\gamma,R} |z|. \label{h4-est-2nd}
\end{align}
As $z$ was arbitrary, \eqref{h4-est-2nd} holds for all $z\in\C$ with $|z|>3R$.
Using \eqref{h4-est-2nd} in \eqref{h4-est-1st} thus yields
\begin{equation}
|h_4^{\mrm{num}}(z)|\leq e^{c_4^{\mrm{num}} |z|}\label{h4-est-4th}
\end{equation}
for some $c_4^{\mrm{num}}=c_4^{\mrm{num}}(\gamma,R)>0$ and all $z\in\C$ with $|z|>3R$.

The quantity $h_4^{\mrm{den}}$ is bounded from below in a similar fashion:
\begin{align}
|h_4^{\mrm{den}}(z)|&=\prod_{\substack{(m,n)\in\m{I}' \\ |\lambda_{m,n}-z|>2R \\|\omega_{m,n}|>2R \\|\lambda_{m,n}|\leq 2|z|}}\left|1-\frac{\omega_{m,n}-\lambda_{m,n}}{\omega_{m,n}}\right|\geq \prod_{\substack{(m,n)\in\Z^2 \\ 2R<|\omega_{m,n}|\leq 2|z|+R}}\left(1-\frac{R}{|\omega_{m,n}|}\right),\notag\\
&\geq \exp\Bigg[ \sum_{2R<|\omega_{m,n}|\leq 2|z|+R}-\frac{2R}{|\omega_{m,n}|} \Bigg],\label{h4-est-3rd}
\end{align}
where in the last inequality we used $\log(1+x)\geq 2x$ for $x\in\left[-\frac{1}{2},0\right]$. Another integral bound yields
\begin{equation*}
\sum_{2R<|\omega_{m,n}|\leq 2|z|+R}\frac{1}{|\omega_{m,n}|}\lesssim_{\,\gamma,R} |z|,
\end{equation*}
which together with \eqref{h4-est-3rd} gives
\begin{equation}
|h_4^{\mrm{den}}(z)|\geq e^{-c_4^{\mrm{den}}|z|}, 
\label{h4-est-5th}
\end{equation}
for some $c_4^{\mrm{den}}=c_4^{\mrm{den}}(\gamma,R)>0$ and all $z\in\C$ with $|z|>3R$. 
Combining \eqref{h4-est-4th} and \eqref{h4-est-5th} thus yields
\begin{equation}
|h_4(z)|\leq e^{(c_{4}^{\mrm{num}}+c_4^{\mrm{den}})|z|},\label{h4-est-6th}
\end{equation}
for all $z\in\C$ with $|z|>3R$. 

\noindent \textit{Bounding $h_5$}:  Note that $|\lambda_{m,n}|>2|z|$ implies $|\lambda_{m,n}|>2|z|> 6R$, and so $\{(m,n)\in\m{I}': |\lambda_{m,n}|>2|z|\}=A(z)$. We thus have $h_5=h_{\mrm{aux}}$, which satisfies \eqref{eq:haux-final-bd-1}. 

\noindent \textit{Bounding $h$}: We combine \eqref{h1-final-est}, \eqref{eq:h2-final-est}, \eqref{eq:h3-only-est},  \eqref{h4-est-6th}, and \eqref{eq:haux-final-bd-1} to obtain
\begin{equation}\label{eq:h-final-est-case1}
|h(z)|\lesssim_{\,s,\theta,\gamma,R} \rho^{-1}  e^{(1+c_{4}^{\mrm{num}}+ c_4^{\mrm{den}}+1100R)|z| +  c_1|z|\log|z|}\leq \rho^{-1}  e^{d |z|\log|z|},
\end{equation}
for some $d=d(s,\theta,\gamma, R)>0$ and all $z\in\C$ with $|z|>3R$. This completes the derivation of the desired upper bound on $h$ in the case $|z|> 3R$.

\noindent \textit{Case $|z|\leq 3R$}: We write $h(z)=h_6(z)h_7(z)$, where
\begin{equation*}
h_6(z)=\frac{d(z,\Omega_\gamma)}{d(z,\Lambda)} \prod_{\substack{(m,n)\in\m{I}'\\|\lambda _{m,n}|\leq 6R}}\frac{1-z/\lambda_{m,n}}{1-z/\omega_{m,n}},\quad\text{and}\quad h_7(z)=\hspace*{-3mm}\prod_{\substack{(m,n)\in\m{I}'\\|\lambda_{m,n}|>6R}}\frac{(1-z/\lambda_{m,n})\exp(z/\lambda_{m,n})}{(1-z/\omega_{m,n})\exp(z/\omega_{m,n})}.
\end{equation*}
Note that $|\lambda_{m,n}|>6R$ and $|z|\leq 3R$ together imply $|\lambda_{m,n}|>2|z|$, and so $h_7=h_{\mrm{aux}}$. Hence it only remains to bound $h_6$. To this end, write $|h_6|=h_6^{\mrm{num}}/h_6^{\mrm{den}}$, where
\begin{equation*}
h_6^{\mrm{num}}(z)=\frac{z}{d(z,\Lambda)}\prod_{\substack{(m,n)\in\m{I}'\\|\lambda_{m,n}|\leq 6R}}\frac{\lambda_{m,n}-z}{\lambda_{m,n}}\quad\text{and}\quad h_6^{\mrm{den}}(z)=\frac{z}{d(z,\Omega_\gamma)}\prod_{\substack{(m,n)\in\m{I}'\\|\lambda_{m,n}|\leq 6R}}\frac{\omega_{m,n}-z}{\omega_{m,n}}.
\end{equation*}
Now, the term $d(z,\Lambda)$ cancels with either $z$ or one of the factors $\lambda_{m,n}-z$, and similarly, $d(z,\Omega_\gamma)$ cancels with either $z$ or one of the factors $\omega_{m,n}-z$. In either case the numerators of the terms remaining in the product satisfy $|\omega_{m,n}-z|\geq\frac{\gamma}{2}$. We again recall that $|\lambda_{m,n}|> \frac{s}{2}$ for all $(m,n)\in\m{I}'\setminus\m{I}_s $, and there is at most one $(m',n')\in \m{I}_s\setminus\{(0,0)\}$, and for this $(m',n')$ we have $|\lambda_{m',n'}|\geq \rho$.
These observations together 
 yield the following bounds:
\begin{align}
|h_{6}^{\mrm{num}}(z)|&\leq 3R \rho^{-1}\prod_{\substack{(m,n)\in\m{I}'\\|\lambda_{m,n}|\leq 6R}} \frac {(|\lambda_{m,n}|+ |z|)\wedge 1}{(s/2)\vee 1 } \leq 3R  \rho^{-1}\left( \frac {(9R) \wedge 1}{(s/2)\vee 1 } \right)^{144\theta R^2}\notag\\
|h_6^{\mrm{den}}(z)|&\geq\prod_{\substack{(m,n)\in\m{I}'\\|\lambda_{m,n}|\leq 6R}} \frac{(\gamma/2)\wedge 1}{|\omega_{m,n}-\lambda_{m,n}|+|\lambda_{m,n}|}\geq \left(\frac{(\gamma/2)\wedge 1}{7R}\right)^{144\theta R^2}, \quad \text{for } z\in\C \text{ s.t. }|z|\leq 3R. \notag
\end{align}
Therefore $|h_6(z)|\lesssim_{\,s,\theta,\gamma,R} \rho^{-1}$, for $z$ with $|z|\leq 3R$, which together with \eqref{eq:haux-final-bd-1} gives
\begin{equation}\label{eq:h-final-est-case2}
|h(z)| \lesssim_{\,s,\theta,\gamma,R} \rho^{-1} e^{1100 R}|z|,\quad \text{for } z\in\C \text{ s.t. }|z|\leq 3R. 
\end{equation}
The inequalities \eqref{eq:h-final-est-case1} and \eqref{eq:h-final-est-case2} can now be combined to yield the bound \eqref{eq:h-goal-bound}, concluding the proof.
 
\end{proof}

\begin{proof}[Proof of Lemma \ref{STFT=Barg-lem}]
Consider first the case when $y \in \m{S}(\R)$ is a Schwartz function. We then have $y\in L^2(\R)$, and thus
\begin{align}
\innerProd{y}{\pi(\lambda)\varphi}_{\moduPair{q}{p}}&= \innerProd{\m{V}_{\varphi}y}{\m{V}_\varphi \pi(\lambda)\varphi}_{L^q(\R^2)\times L^p(\R^2)} \notag\\
&=\iint_{\R^2} (\m{V}_{\varphi}y)(s, \xi) \;\overline{(\m{V}_{\varphi}\pi(\lambda)\varphi)(s, \xi)}\; \mrm{d}s\mrm{d}\xi \notag \\
&=\int_\R y(t)\overline{\left(\pi(\lambda)\varphi\right) (t)}\;\mrm{d}t\label{STFTorthrel}\\
&=(\m{V}_\varphi y)(\lambda)\notag\\
&=e^{-\pi i \tau\nu}e^{-\pi |\lambda|^2/2}(\bargmann y)(\bar{\lambda}),\label{GroSTFT=Barg}
\end{align}
where \eqref{STFTorthrel} follows from \cite[Thm. 3.2.1]{Groechenig2000}, and \eqref{GroSTFT=Barg} is \cite[Prop. 3.4.1]{Groechenig2000}.

Now take an arbitrary $y\in\MSp{q}{\R}$. As $\m{S}(\R)$ is dense in $\MSp{q}{\R}$ for $q\in[1,\infty)$ (see \cite[Prop. 11.3.4]{Groechenig2000}), we can take a sequence $\{y_n\}_{n=1}^{\infty}\subset\m{S}(\R)$ such that $y_n\to y$ in $\MSp{q}{\R}$. The calculation above thus shows that
\begin{equation}\label{STFT=Barg argument 1}
\langle y_n,\pi(\lambda)\varphi \rangle_{M^q(\R)\times M^p(\R)}=e^{-\pi i \tau\nu}e^{-\pi |\lambda|^2/2}(\bargmann y_n)(\bar{\lambda}),\quad\forall n\in\N.
\end{equation}
Furthermore, as the dual pairing is continuous, we have
\begin{equation*}
\langle y_n,\pi(\lambda)\varphi\rangle_{M^q(\R)\times M^p(\R)} \to \langle y,\pi(\lambda)\varphi\rangle_{M^q(\R)\times M^p(\R)}\quad\text{as }n\to\infty.
\end{equation*}
On the other hand, by the isometry property \eqref{eq:Btrans-isom} we also have $\|\bargmann y_n-\bargmann y\|_{\m{F}^q(\C)}=\|y_n-y\|_{M^q(\R)}\to 0$ as $n\to \infty$. Thus, as the evaluation functional $F\mapsto F(\overline{\lambda})$ is continuous on $\m{F}^q(\C)$ (see \cite[Lem. 2.32]{Kehe2012}), we obtain $(\bargmann  f_n)(\overline{\lambda})\to (\bargmann f)(\overline{\lambda})$, which together with \eqref{STFT=Barg argument 1} and \eqref{GroSTFT=Barg} establishes the claim of the proposition.
\end{proof}

\begin{proof}[Proof of Lemma \ref{ShortTheoryLemma}]
(i) Suppose that $\m{A}$ is bounded below. Then the operator $\tilde{\m{A}}:X\to\mrm{Im}(\m{A})$ given by $\tilde{\m{A}}(x)=\m{A}(x)$, for $x\in X$, is a continuous map between Banach spaces, and has a continuous inverse. In other words, $\tilde{\m{A}}$ is an isomorphism between Banach spaces. Thus $\tilde{\m{A}}^*:(\mrm{Im}(\m{A}))^*\to X^*$ is also an isomorphism between Banach spaces, and so, by the inverse mapping theorem \cite[Cor. 2.12]{Rudin1991}, so is $(\tilde{\m{A}}^*)^{-1}:X^*\to(\mrm{Im} (\m{A}))^*$.  Consider now an arbitrary $f\in X^*$, and set $h=(\tilde{\m{A}}^*)^{-1}f$. As $h$ is a continuous linear functional on $\mrm{Im}(\m{A})\subset Y$, it follows by the Hahn-Banach theorem \cite[Thm. 3.6]{Rudin1991} that $h$ can be extended to a continuous linear functional $h_Y$ defined on $Y$. Now, since $h_Y\!\mid_{\mrm{Im}(\m{A})}=h$, we have
\begin{equation*}
\begin{aligned}
\langle \m{A}^* h_Y, x\rangle&=\langle  h_Y , \underbrace{\m{A}x}_{\mathclap{\in\mrm{Im}(\m{A})}}\rangle=\langle h,\m{A}x\rangle=\langle h,\tilde{\m{A}}x\rangle=\\
&=\langle \tilde{\m{A}}^*h, x\rangle=\langle f,x\rangle\quad\text{for }x\in X,
\end{aligned}
\end{equation*}
and thus, as $x\in X$ was arbitrary, we deduce that $\m{A}^* h_Y=f$. Finally, since $f$ was arbitrary, we have that $\m{A}^*$ is surjective.

\noindent(ii) Let $f$ be an arbitrary element of $X^*$ with $\|f\|=1$, and let $g\in Y^*$ be such that $\m{A}^* g=f$ and $a\|g\|\leq 1$. Note that then $g\neq 0$, and so $g/\|g\|$ is a well-defined element of $Y^*$ of unit norm. Therefore,
\begin{equation*}
\|\m{A}x\|\geq \left|\left< \m{A}x,\frac{g}{\|g\|}\right>\right|=\|g\|^{-1}\left|\langle x,\m{A}^* g\rangle\right|\geq a\left|\langle x,f\rangle\right|,
\end{equation*}
for all $x\in X$. Taking the supremum of the right-hand side over $f\in X^*$ and using the fact that  $\sup_{f\in X^*,\|f\|=1}|\langle x, f\rangle|=\|x\|$ yields $\|\m{A}x\|\geq a\|x\|$, as desired.
\end{proof}

\begin{proof}[Proof of Lemma \ref{IrregConvLem}]
For $m,n\in\Z$ write 
\begin{equation*}
\textstyle K_{m,n}=\left[\frac{s}{\sqrt{2}} \left(m_{\lambda}-\frac{1}{2}\right),\frac{s}{\sqrt{2}} \left(m_{\lambda}+\frac{1}{2}\right)\right)\times\left[\frac{s}{\sqrt{2}} \left(n_{\lambda}-\frac{1}{2}\right),\frac{s}{\sqrt{2}} \left(n_{\lambda}+\frac{1}{2}\right)\right),
\end{equation*}
and, for $\lambda\in\Lambda$, let $(m_\lambda,n_\lambda)$ be the (unique) element of $\Z^2$ such that $\lambda\in K_{m,n}$. Note that, as $\sep(\Lambda)=s$, every $K_{m,n}$ contains at most one element of $\Lambda$. Next, define the functions
\begin{equation*}
A(z)=\frac{2}{s}  \sum_{\lambda \in\Lambda }a_{\lambda} \mathds{1}_{K_{m_\lambda,n_\lambda}}(z)\quad \text{and}\quad \tilde{f}(z)=\max_{|w|\leq s}|f(z+w)|.
\end{equation*}
We then have
\begin{equation*}
\sum_{\lambda\in\Lambda }|\theta_\lambda||f(z-\lambda)|\leq \sum_{\lambda\in\Lambda} |\theta_{\lambda}| \;\tilde{f}\left(z-{s(m_\lambda+in_\lambda)}/{\sqrt{2}} \right)\leq \int_{\C}A(w)\tilde{f}(z-w)|\mrm{d}w| = (A\ast \tilde{f})(z),
\end{equation*}
for all $z\in\C$.
Therefore, using \cite[Prop. 11.1.3]{Groechenig2000}, we obtain
\begin{equation*}
\|A\ast \tilde{f}\|_{L^q(\C)}\leq \|\tilde{f}\|_{L^1(\C)} \|A\|_{L^q(\C)}\lesssim s^{-2/p}\|f\|_{W(L^\infty,L^1)} \|\theta\|_{\ell^p(\Lambda)},
\end{equation*}
where the last equality follows by computing the norm of $A$ explicitly. 
\end{proof}

\begin{proof}[Proof of Lemma \ref{RecursiveSquareLemma}]
Note that it suffices to prove the claim for $j=n-1$, as the general statement then follows by induction. To this end, divide $K_n$ into four disjoint squares of side length $\sqrt{2}\left(2^{n-1}+\frac{1}{2}\right)$. By the pigeonhole principle, one of these squares must contain at least $2^{2(n-1)}+1$ points of $K_n\cap Y$. Denote this square by $K'$, and Let $K_{n-1}$ be the square which contains $K'$ and satisfies property (i) in the statement of the Lemma. Then $\# (K_{n-1}\cap Y)\geq \# (K'\cap Y)\geq 2^{2(n-1)}+1$, and so $K_{n-1}$ satisfies (ii), as desired.
\end{proof}

\end{document}